\documentclass[11pt, a4paper, reqno]{article}
\usepackage{amssymb}
\usepackage{amsmath}
\usepackage{amsthm,thmtools}
\usepackage{mathtools}
\usepackage{mathrsfs}
\usepackage[T1]{fontenc}
\usepackage[dvipsnames]{xcolor}
\usepackage{combelow}
\usepackage{enumitem}
\setenumerate{label={\rm (\alph{*})}}
\usepackage[
    colorlinks=true,
    citecolor=PineGreen,
    urlcolor=CadetBlue,
    linkcolor=RedViolet]{hyperref}
\usepackage[nameinlink,capitalize]{cleveref}
\usepackage{graphicx}
\usepackage{tikz}
\usepackage{tikz-cd}
\usepackage{pgfplots}
\pgfplotsset{compat=newest}
\usepackage[margin=1cm]{caption}
\usepackage{geometry}
\usepackage[backend=biber,style=alphabetic, 
maxalphanames = 4, 
  minalphanames = 3, 
  maxbibnames=99, 
]{biblatex}

\usepackage{algorithm}
\usepackage[noend]{algpseudocode}
\algrenewcommand{\algorithmicrequire}{\textsc{Input:}}
\algrenewcommand{\algorithmicensure}{\textsc{Output:}}

\usepackage{soul}
\usepackage{verbatim}
\usepackage{subcaption}

\numberwithin{equation}{section}
\newtheorem{theorem}{Theorem}[section]
\newtheorem{lemma}[theorem]{Lemma}
\newtheorem{proposition}[theorem]{Proposition}
\newtheorem{corollary}[theorem]{Corollary}
\theoremstyle{definition}
\newtheorem{definition}[theorem]{Definition}

\newtheorem{examplex}[theorem]{Example}
\newenvironment{example}
{\pushQED{\qed}\examplex}
{\popQED\endexamplex}
\theoremstyle{remark}
\newtheorem{remark}[theorem]{Remark}

\newcommand{\conv}{\operatorname{conv}}
\newcommand{\R}{\mathbb{R}}

\newcommand{\Q}{\mathbb{Q}}
\newcommand{\Sph}{\mathbb{S}}

\newcommand{\vol}{\operatorname{vol}}
\newcommand{\verts}{\operatorname{vert}}

\newcommand{\gr}{\operatorname{Gr}}
\newcommand{\graff}{\operatorname{Gr}_{\mathrm{aff}}}
\DeclareMathOperator{\depth}{depth}
\DeclareMathOperator{\capvol}{CapVol}
\newcommand{\HS}{\mathcal{HS}}

\newcommand{\TFNP}{\textsf{TFNP}}
\newcommand{\PPA}{\textsf{PPA}}

\newcommand{\AlphaHS}{\ensuremath{\alpha}\textsc{-HS}}
\newcommand{\HamSandwich}{\textsc{Ham-Sandwich}}

\newcommand\blfootnote[1]{%
  \begingroup
  \renewcommand\thefootnote{}\footnote{#1}%
  \addtocounter{footnote}{-1}%
  \endgroup
}

\title{\vspace{-1cm}\textbf{\textsc{From Ham-Sandwich to Centerpoints:\\ Semialgebraic Algorithms for Cutting Polytopal Measures}}}
\author{Marie-Charlotte Brandenburg, Jes\'us A. De Loera, and Chiara Meroni}
\date{}

\begin{document}

\maketitle
\begin{abstract}
We design exact algorithms for the ham-sandwich and centerpoint theorems for polytopal measures. Our key observation is that the cap-volume function of such a measure, i.e., the volume cut off by a halfspace, is piecewise rational on a natural decomposition of the space of oriented hyperplanes. This lets us recast prescribed-proportion cutting problems as semialgebraic feasibility problems. For fixed ambient dimension, this yields polynomial-time algorithms to decide the existence of cuts, describe the full solution set, and sample or enumerate solutions. We extend this framework to the center transversal theorem, showing that spaces of deep affine flats are semialgebraic, which holds for centerpoints. We further show that the set of centerpoints of a convex polytope coincides with its floating body at level $1/(d+1)$, a useful semialgebraic description.
\end{abstract}

\noindent
\blfootnote{
\textbf{MSC classes:} {Primary 52B55;  
Secondary 52A35,  
52A38,  
14P10,  
14Q30,  
68Q25,  
68W30.}  
}

\vspace{-0.5cm}
\section{Introduction}
In this article, we develop computational versions of 
the \emph{ham-sandwich theorem} \cite{StoneTukey1942} and the \emph{centerpoint theorem} \cite{Rado46:TheoremGeneralMeasure}. These classical results establish the existence of hyperplanes that partition measures, or point sets, into balanced pieces. Together, they underlie a remarkably rich line of work in computational geometry on efficient equipartitioning, with exciting direct applications to economics, fair-division problems, gerrymandering, and statistics \cite{BramsTaylor1996,Su1999,Soberon2017,filosratsikasGoldberg2023,Fojtik02042024}.

In this paper, we discuss algorithmic variations of these two classical problems. The difficulty is already visible in the ham-sandwich setting illustrated in \Cref{fig:intro}: given three ingredients, say bread, ham, and cheese, can one cut the sandwich with a single plane so that each ingredient is divided exactly in half? A key contribution of this paper is an algorithm that solves this tasty sandwich problem and finds all possible fair bisections; see \Cref{fig:intro}.

\begin{figure}[hb]
    \centering
    \includegraphics[height=4.0cm]{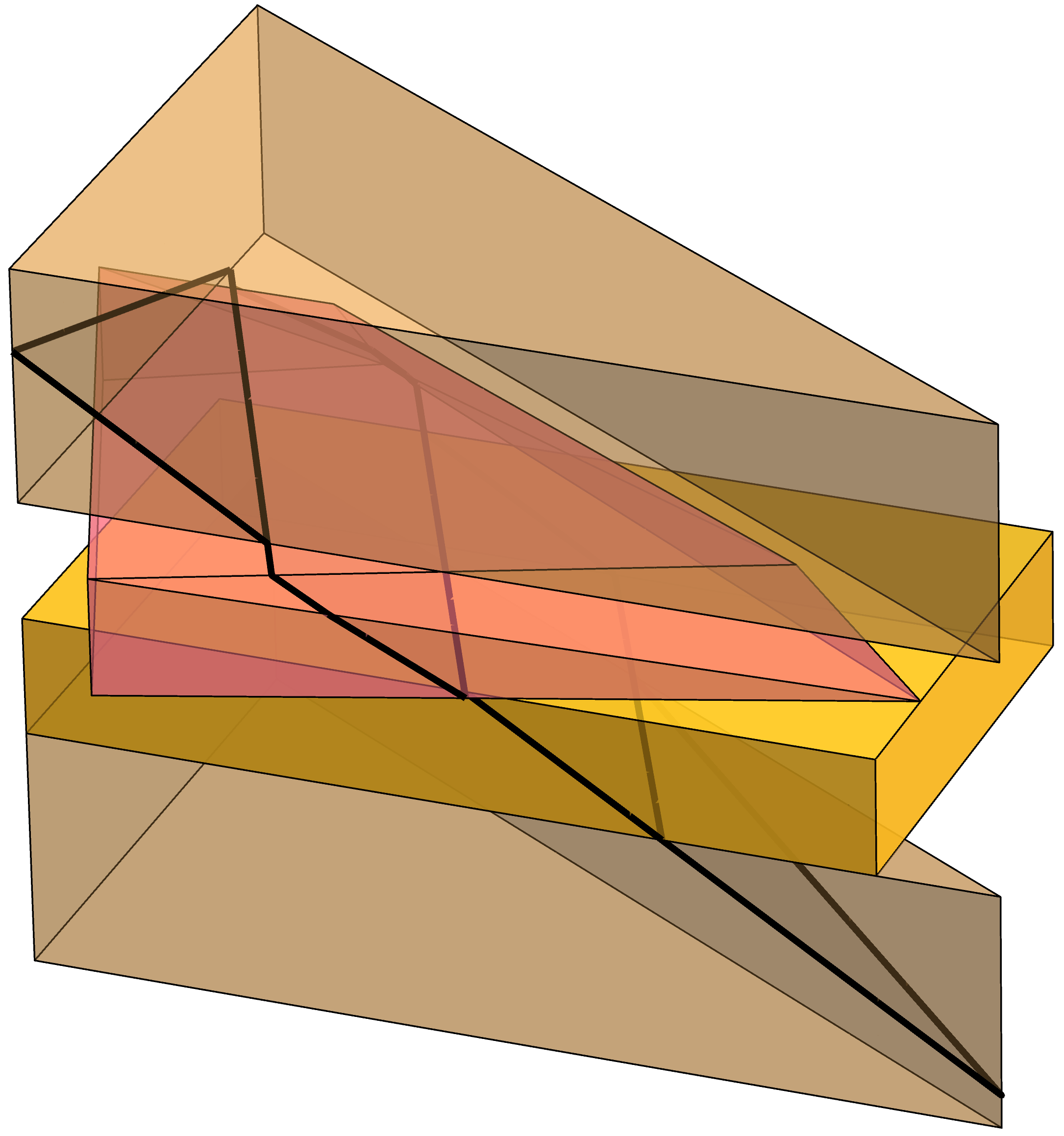} 
    \qquad\includegraphics[height=4.0cm]{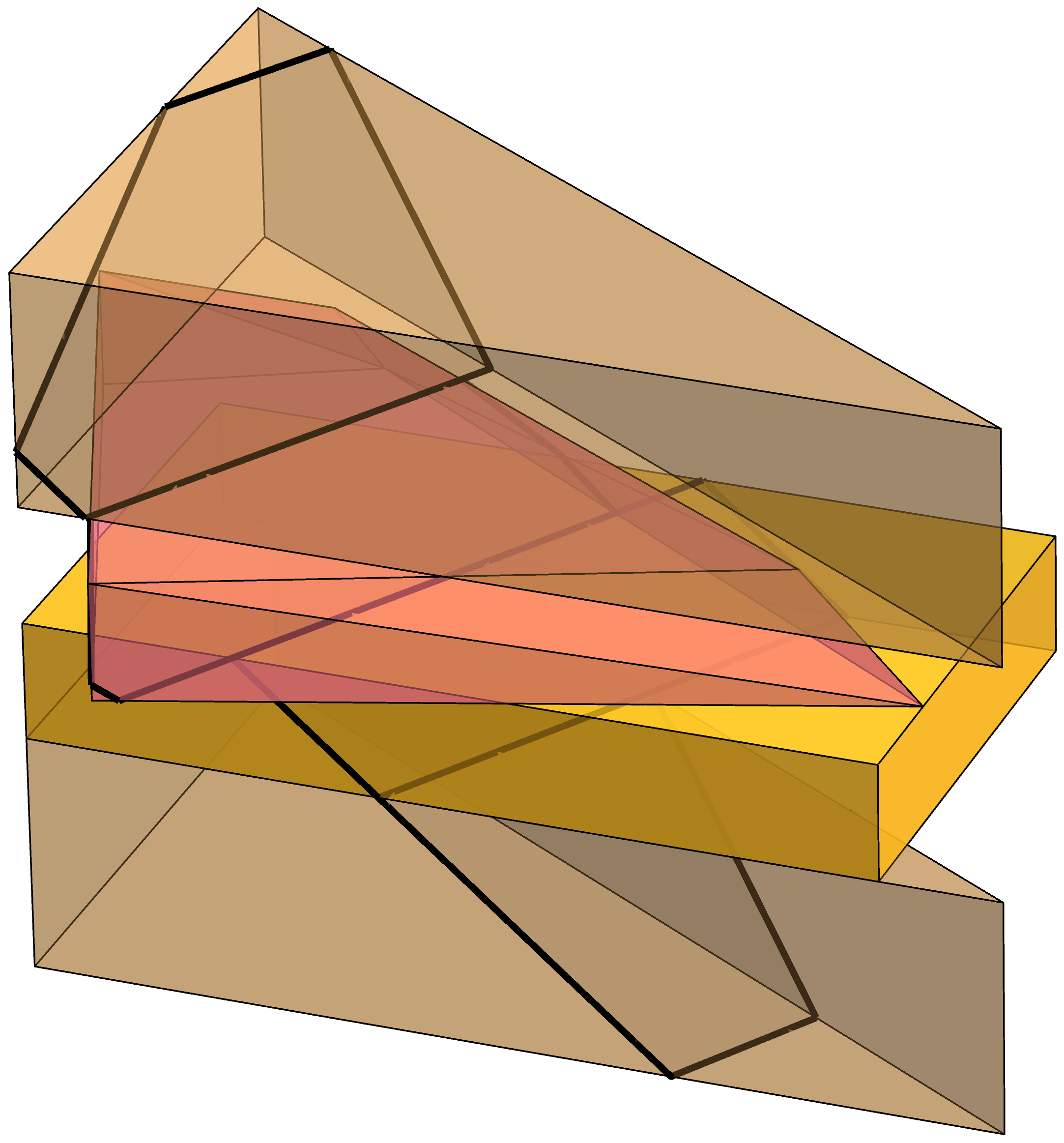}
    \qquad\includegraphics[height=4.0cm]{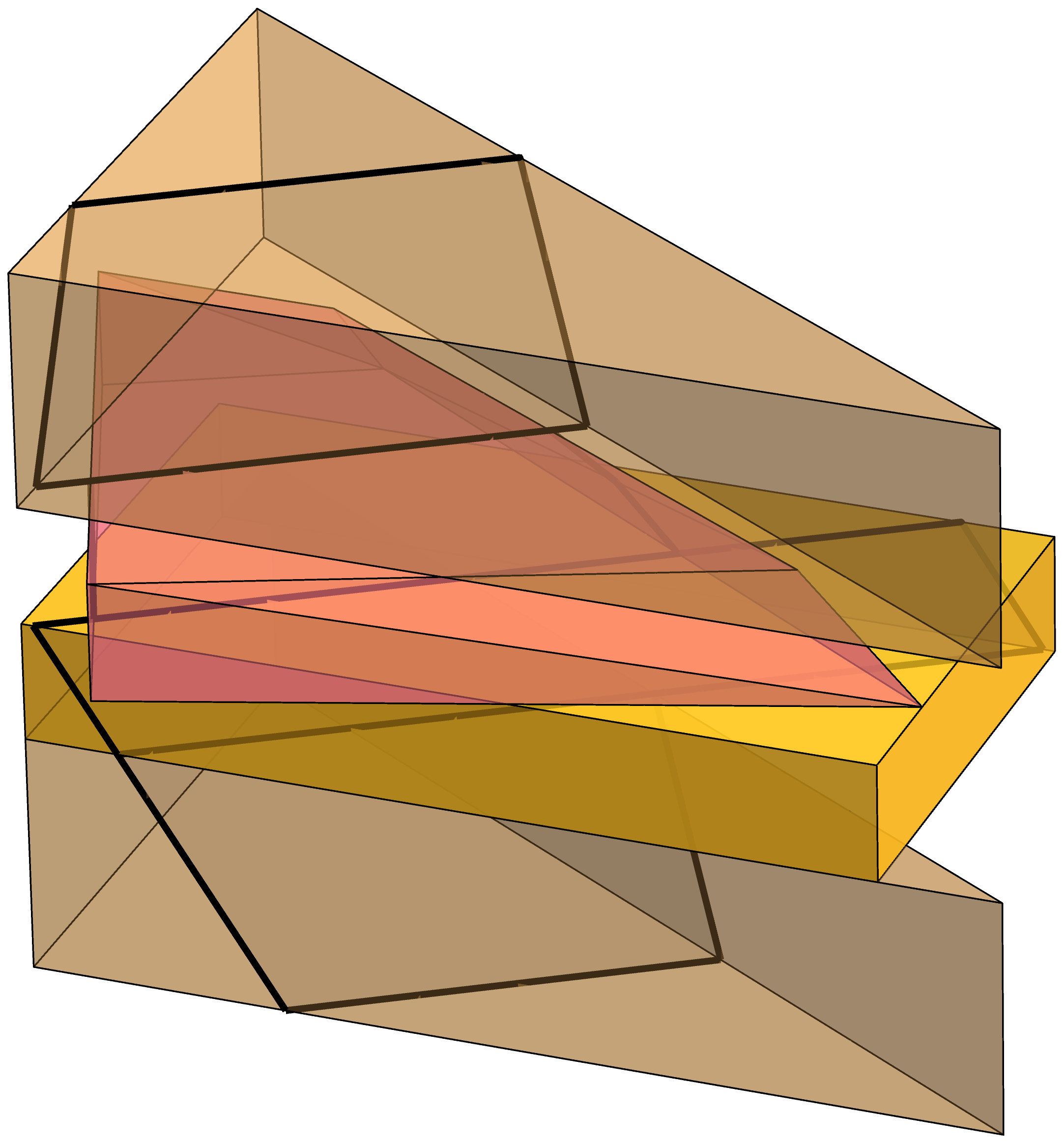}
    \caption{All possible ways to cut this ham-cheese-sandwich into two equal parts.}
    \label{fig:intro}
\end{figure}  

The existence of such bisecting cuts follows from famous and elegant proofs, but these proofs are typically non-constructive, relying on tools from algebraic topology or convex geometry. Consequently, many partition results remain mysterious from the point of view of computation. Here we present algorithms for both the ham-sandwich theorem and the centerpoint theorem, as well as for several of their generalizations. Our algorithms apply when the input measures are finite unions of convex polytopes. More precisely:
\begin{definition}
    We call a measure $P$ in $\R^d$ \emph{polytopal} if it is the restriction of the Lebesgue measure to a union of finitely many full-dimensional convex polytopes $Q_{1},\ldots, Q_{n} \subset \R^d$. We refer to $Q_j$ as subpolytopes of $P$. We say that such a measure is \emph{rational} if the vertices of all $Q_i$ have rational coordinates.
\end{definition}

Our first key observation is that for rational polytopal measures one can compute ham-sandwich cuts and centerpoints exactly, using techniques from computational semialgebraic geometry \cite{BasuPollackRoy2006,BochnakCosteRoy1998}. For polytopal measures, the sets of bisecting hyperplanes are natural semialgebraic sets, and, in fact, most problems of partitioning polytopes or polytopal measures can also be reframed as feasibility problems in semialgebraic geometry \cite{BasuPollackRoy2006}. For all our algorithmic results, the input size includes the number of vertices and the bit-size of their coordinates. Throughout the algorithmic complexity statements, the ambient dimension \(d\) is fixed unless explicitly stated otherwise. Outputs are semialgebraic descriptions or real algebraic sample points. In the rest of the paper we write polytopal measure to mean rational polytopal measure.

\subsection{Ham-Sandwich Contributions}
The \emph{ham-sandwich theorem} states that for any $d$ Lebesgue-measurable sets $S_1, \ldots, S_d \subset \R^d$ of finite measure, there exists an oriented hyperplane~$H$ that
simultaneously bisects them: each of the two open halfspaces determined by $H$ contains exactly half the
measure of each $S_i$. This is also sometimes called the equipartition problem. The classical proof of Stone and Tukey~\cite{StoneTukey1942} proceeds via the Borsuk--Ulam theorem: After parameterizing oriented directions by the sphere $\Sph^{d-1}$, one constructs an antipodal map whose zeros correspond to simultaneous bisections. See, e.g., \cite{Matousek2003} for details.

The ham-sandwich theorem is a cornerstone of combinatorial and convex geometry, but its elegant existence proof leaves open a natural computational challenge: \emph{how hard is it to find a bisecting hyperplane?}  We denote the corresponding computational problem by \HamSandwich{}: given the sets $S_1,\ldots,S_d$ as input, find a hyperplane that simultaneously bisects them.
We call such a hyperplane a \emph{ham-sandwich cut}, or simply an \emph{HS cut}. 

The classical ham-sandwich theorem requires a \emph{bisection}: exactly half of each set on each side. A natural and more flexible question is whether one can prescribe arbitrary proportions instead. Bárány, Hubard, and Jerónimo~\cite[Theorem 2]{BaranyHubardJeronimo2008} showed that if $S_1, \ldots, S_d \subset \R^d$ are convex, well-separated sets (i.e., any sub-collection can be strictly separated from the remaining sets by a hyperplane), then for any target fractions $\alpha_1, \ldots, \alpha_d \in (0,1)$ there is a \emph{unique} oriented hyperplane $H$ such that the measure of $S_i$ lying in the negative halfspace is equal to $\alpha_i \cdot \mu(S_i)$ for every $i$. Their proof uses Brouwer's Fixed-Point Theorem, while uniqueness follows from the well-separation condition. Steiger and Zhao~\cite{SteigerZhao2010} studied the discrete case, where one has $d$ finite, well-separated point sets $P_1, \ldots, P_d \subset \R^d$ in weak general position, and target integers $\alpha_i \in [|P_i|]$. They showed that there exists a unique oriented hyperplane $H$ passing through one point of each $P_i$ such that there are $\alpha_i$ points of $P_i$ in the negative halfspace, for every $i \in [d]$. We refer to these continuous and discrete results as the \emph{$\alpha$-ham-sandwich theorem}, and to the corresponding computational search problem as \AlphaHS{}.
More generally, given finite measures $\mu = \{\mu_1,\ldots,\mu_k\}$ supported on $\R^d$, we say that a halfspace $H\subset\R^d$ is an \emph{$\alpha$-cut} of $\mu$ if 
\[
\mu_i(H)=\alpha_i\,\mu_i(\mathbb R^d)
\qquad \text{for all } i=1,\ldots,k.
\]
Our first main result is the following.

\begin{theorem}[HS algorithm]\label{thm:HS_1/2_intro}
    Let $\mathcal{P} = \{P_1,\ldots,P_d\}$
    be a $d$-tuple of rational polytopal measures on $\R^d$. There exists an algorithm that takes as input the vertices of the subpolytopes of all $P_i$ and outputs the set of all HS cuts of $\mathcal{P}$. For fixed $d$, the complexity of the algorithm is polynomial in the total number of input vertices.
\end{theorem}
More precisely, the algorithm outputs a quantifier-free semialgebraic description of all HS cuts of $\mathcal{P}$. From this description, the algorithm can compute its dimension, decide whether it is empty, and produce one sample point in each connected component. In zero-dimensional instances, it outputs all HS cuts, represented by real algebraic numbers.
The same approach applies more generally to arbitrary target proportions and to any number of measures.
In particular, when the number of measures exceeds the ambient dimension, an $\alpha$-cut need not exist; nevertheless, our algorithm decides existence and, when cuts do exist, describes the entire space of them.

\begin{theorem}[\AlphaHS{} algorithm]\label{thm:generalized_HS_intro}
    Let $\mathcal{P} = \{P_1,\ldots,P_k\}$ be a $k$-tuple of rational polytopal measures on $\R^d$. Let $\alpha_1,\dots,\alpha_k \in (0,1)\cap \Q$. There exists an algorithm that takes as input the vertices of the subpolytopes of all $P_i$ and outputs the set of all $\alpha$-cuts of $\mathcal{P}$.
    For fixed $d$, the complexity of the algorithm is polynomial in the total number of input vertices.
\end{theorem}

\subsection{Centerpoint Contributions}
The second classical result we study is the centerpoint theorem \cite{Rado46:TheoremGeneralMeasure}. Let $\mu$ be a finite measure on $\R^d$; in the discrete setting, one may take $\mu$ to be the uniform measure over a finite point set. For a point $x\in\R^d$, the \emph{Tukey halfspace depth} of $x$ with respect to $\mu$ is
\[
\mathrm{depth}_\mu(x) = \inf\{\mu(H_{\geq 0}) \mid  H_{\geq 0}\subset\R^d \text{ is a closed halfspace with } x\in H_{\geq 0}\}.
\]
A \emph{centerpoint}, also called a \emph{Tukey point}, is any point $x^\star\in\R^d$ satisfying
\[
\operatorname{depth}_\mu(x^\star)\ge \frac{\mu(\mathbb R^d)}{d+1}.
\]
The existence of such a point is guaranteed by the \emph{centerpoint theorem}: For every finite Borel measure $\mu$ on $\R^d$, and in particular for every finite point set, there exists a point $x\in\R^d$ such that every closed halfspace containing $x$ has $\mu$-mass at least $\mu(\R^d)/(d+1)$. 

Centerpoints are fundamental objects in discrete and computational geometry, optimization, and robust statistics. They provide a high-dimensional analogue of the median, are invariant under affine transformations, and yield provably robust representatives of data, with breakdown point at least $1/(d+1)$.
They also serve as key tools in geometric algorithms and existence proofs, including Helly-type theorems, weak $\varepsilon$-nets, and various partition results. We refer to \cite{DeLoera2019,Matousek02:LecturesDiscreteGeom} for further details and applications.

Already when $\mu$ is the uniform measure on a finite set $X\subset\R^d$ of $n$ points, computing a centerpoint is algorithmically nontrivial.
In fixed dimension $d$, the depth function is piecewise constant on a hyperplane arrangement determined by $X$, and exact computation can be performed in polynomial time in $n$, although typically an exponential in $d$. Our algorithm fits into this general framework and extends it from point sets to polytopal measures.

\begin{theorem}[Centerpoint algorithm]\label{thm:centerpoints_intro}
   Let $P \subset \R^d$ be a rational polytopal measure. There exists an algorithm that takes as input the vertices of all subpolytopes defining $P$ and outputs at least one centerpoint for $P$. For fixed dimension, the complexity of the algorithm is polynomial in the total number of input vertices. Moreover, the set of all centerpoints of $P$ is semialgebraic, and a defining semialgebraic description can be recovered in polynomial time.
\end{theorem}

The geometric starting point of our algorithm is that the set of centerpoints of a convex polytopes can be viewed as a \emph{convex floating body} \cite{BarLar88:Floating,SchWer90:Floating,MorWer19:FloatingPolytopes}.
Informally, these bodies capture the deep core of a convex body after removing intersections with halfspaces cutting away a prescribed proportion of volume. We show that, for polytopes, these objects admit natural semialgebraic descriptions and can therefore be computed using the same tools of real algebraic geometry that underlie our ham-sandwich algorithms.

More generally, this perspective leads to a common framework for studying not only centerpoints but also higher-dimensional analogues of depth. We prove in \Cref{thm:deep-flats-semi-alg} that, in fact, the intersections with flats are also semialgebraic sets over the affine Grassmannian.
This gives a common algorithmic framework for results interpolating between the ham-sandwich and centerpoint theorems, such as the center transversal theorem \cite{Dolnikov92:GeneralizationHSCenterpoint,ZivVre90:ExtensionHS}.

\paragraph{Organization}
Our paper is organized as follows. We start by reviewing prior computational work in \Cref{sec:prior}, and in \Cref{sec:semialg} we recall real algebraic geometry tools and algorithms used throughout. \Cref{sec:HS} contains the main construction used in the algorithm for hyperplane cuts of polytopal measures. The proofs of \Cref{thm:generalized_HS_intro} and \Cref{thm:HS_1/2_intro} are given in this section. In \Cref{sec:transversal}, our real algebraic geometry setup allows us to algorithmically interpolate from the ham-sandwich theorem to the centerpoint theorem via the center transversal theorem, see \Cref{thm:deep-flats-semi-alg} and \Cref{prop:deep-flats-algorithm}. Finally, \cref{sec:floating-bodies} develops the connection between centerpoints of convex polytopes and convex floating bodies.

\paragraph{Acknowledgments} We thank Melissa Daniele for contributions to the code behind this project. MB is supported by the SPP 2458 ``Combinatorial Synergies'', funded by the Deutsche Forschungsgemeinschaft (DFG, German Research Foundation). JDL is grateful for the financial support received through 
NSF grants 2348578 and 2434665. CM is supported by Dr. Max R\"ossler, the Walter Haefner Foundation, and the ETH Z\"urich Foundation.

\section{Prior Algorithmic Work} \label{sec:prior}

What is known about the computational complexity of these problems? We summarize this below. However, most existing work on concrete algorithms has focused on discrete measures (clouds of points), and only a few works discuss polygons.

\paragraph{Complexity of ham-sandwich problems}
The complexity of HS problems varies significantly with the dimension and the type of constraints. In fixed dimension, HS cuts for polygons \cite{EdelsbrunnerWaupotitsch1986,Stojmenovic91} and finite point sets \cite{LoMatousekSteiger1994} are solvable in polynomial time. Steiger and Zhao \cite{SteigerZhao2010} provided an $O(n(\log n)^{d-3})$ algorithm, later refined by Bereg \cite{Bereg2012} to $n \cdot 2^{O(d)}$. This remains linear in the input size $n$ but exponential in $d$.
 
On the other hand, equipartitioning becomes NP-hard when the number of measures $k$ exceeds the dimension $d$, or when multiple simultaneous hyperplanes are required; see \cite{RoldanPensadoSoberon}. 
Because the bisecting hyperplane is guaranteed to exist for \emph{every} valid input, \HamSandwich{}, as a computational problem, belongs to the class $\TFNP$ (Total Function NP) of search problems that always admit a solution~\cite{Papadimitriou1994}.
Since the standard proof routes through the Borsuk--Ulam Theorem, 
whose
discrete version, \emph{Tucker's Lemma}, is $\PPA$-complete~\cite{Aisenberg2020}, one obtains containment of \HamSandwich{} in $\PPA$ for free. The question of completeness remained open for some time, until Filos-Ratsikas and Goldberg \cite{FilosRatsikas2019, filosratsikasGoldberg2023} proved that \HamSandwich{} and its discrete analogs (Necklace Splitting and Consensus-Halving) are indeed $\PPA$-complete \cite{FilosRatsikas2019,filosratsikasGoldberg2023, Aisenberg2020}. While efficient algorithms exist for low dimensions, the gap for general dimension $d$ remains a central challenge, recently addressed by Chiu, Choudhary, and Mulzer \cite{CCM2020}. As we discuss in \Cref{sec:transversal}, one can interpolate between the ham-sandwich theorem and the centerpoint theorem, and Jadhav and Mukhopadhyay \cite{JadhavMukhopadhyay1994} used a ham-sandwich subroutine to compute planar centerpoints.

\paragraph{Complexity of computing centerpoints}
In general, when $d$ is part of the input, computing the Tukey depth and computing a deepest point are known to be computationally hard problems. Testing whether a given point is a centerpoint is coNP-complete in general dimension. Because of this, much work has focused either on fixed-dimensional algorithms or on approximation.
For fixed $d$, exact algorithms for finite point sets run in time polynomial in the number $n$ of points, though with a dependence on the dimension that becomes prohibitive as $d$ grows.
In higher-dimensional settings, randomized and deterministic approximation algorithms compute points with provable halfspace-depth guarantees, typically below the centerpoint threshold $n/(d+1)$ but still sufficient for many applications \cite{ClarksonEppsteinMillerTeng1993,MillerSheehy2010Approximate}.
More recent work gives near-optimal algorithms for several geometric center problems, including Tukey medians and centerpoints, in fixed dimension \cite{ChanHarPeledJones2022OptimalCenters}. Such approximate or fixed-dimensional methods are often useful in applications such as robust estimation and geometric data analysis, where computing exact deepest points may be computationally expensive \cite{statisticians2003DepthContours,Fojtik02042024}.

\section{Key Tools from Real Algebraic Geometry}\label{sec:semialg}
In this paper, we rely on methods from real algebraic geometry. Such tools have already been of great importance in computational geometry and economics (see e.g., \cite{BasuPollackRoy2006,DeLoera2019,Schaeferetal2024} and the reference to applications therein). Here we rely on these techniques for our problems.
As we see later, the sets of all hyperplanes partitioning measures in certain proportions are in fact semialgebraic sets. A \emph{semialgebraic set} in $\mathbb{R}^n$ is a subset defined by finitely many polynomial equalities and inequalities, combined using Boolean operations
(finite unions, finite intersections, and complements). Equivalently, it is any
set obtained from basic sets of the form
\[
\{x\in \mathbb{R}^n : f(x)=0,\ g_1(x)>0,\dots,g_r(x)>0\},
\]
where $f,g_1,\dots,g_r \in \mathbb{R}[x_1,\dots,x_n]$, by finitely many Boolean
operations. Thus, algebraic sets, polyhedra, balls, and many regions bounded by
polynomial curves or surfaces are semialgebraic; see \cite[Chapter~2]{BasuPollackRoy2006}, \cite[Chapter~2]{BochnakCosteRoy1998}, \cite[Chapter~1]{BenedettiRisler1990}.

Semialgebraic sets form a remarkably robust class. They are closed under finite
union, finite intersection, complement, Cartesian product, closure, interior,
and boundary; see \cite[Section~2.3]{BochnakCosteRoy1998} and
\cite[Chapter~2]{BasuPollackRoy2006}. Every semialgebraic set has only finitely
many connected components, each of which is again semialgebraic. More
generally, semialgebraic sets admit finite stratifications into smooth
semialgebraic manifolds and triangulations by finite simplicial complexes
\cite[Chapter~9]{BochnakCosteRoy1998}, \cite[Chapter~2]{BenedettiRisler1990}. Consequently, their topology is tame:
they have finite Betti numbers and a well-defined Euler characteristic.

In what follows, we will make use of \emph{Quantifier Elimination}. Consider a set of the form
\[
S = \left\{ x \in \mathbb{R}^n \colon 
(Q_1 y_1 \in \mathbb{R}) \cdots (Q_k y_k \in \mathbb{R}) \;
\bigvee_{i \in I} \; \bigwedge_{j \in J}
g_{i,j}(x,y)\;\sigma_{i,j}\;0 
\right\}
\]
where each $Q_\ell$ is either $\exists$ or $\forall$ and is called a \emph{quantifier}, $I,J$ are finite sets, all $g_{i,j}$ are polynomials, and $\sigma_{i,j} \in \{=,>,<,\ge,\le,\neq\}$. Then, the quantifier elimination theorem states that $S$ is a semialgebraic set, or in other words, there exists another representation of $S$ which does not require any quantifier $Q_\ell$.
There is no known easy proof for this result, and we refer the reader to \cite[Proposition~2.2.4]{BochnakCosteRoy1998}, \cite[Theorem~2.77]{BasuPollackRoy2006}. 
The theorem is fundamental in real algebraic geometry because it shows that semialgebraic sets remain within the same category under one of the most important operations in mathematics: elimination of variables. It is also the starting point for many algorithmic applications in the area. In particular, we will use it several times in this paper.

\paragraph{Computational Aspects}
In this paper we will be computing with specific semialgebraic sets that arise as sets of all hyperplanes that divide a convex polytope into proportional pieces. We briefly touch on different techniques that may provide concrete algorithms.

There are several approaches from real algebraic geometry that can be applied to actually compute a point inside a semialgebraic set $S$, to show it is empty, or to describe it explicitly. 
A classical method to describe semialgebraic sets $S$ as a union of topological balls is \emph{cylindrical algebraic decomposition} or CAD \cite[Chapter 11]{BasuPollackRoy2006}.
CAD recursively projects and decomposes $\R^n$ into cells where the sign of each polynomial defining $S$ is constant. Inherently, it provides at least one sample point in every connected component of $S$. 
CAD runs in doubly-exponential time in general, and in polynomial time for fixed dimension $d$. However, much of the complexity of CAD lies in the quantifiers describing $S$. In our case, there are no quantifiers, so one can aim for a better bound.

If the interest focuses mainly on samples (such as finding a single ham-sandwich cut), then the so-called \emph{critical point methods}, developed by Grigoriev and Vorobjov \cite{GriVor88:SubexponentialSystems,GriVor92:SubexponentialComponents}, and later improved by many other researchers, provide single-exponential algorithms. One such algorithm can compute a finite set of points that meets every connected component of the solution set.
The critical point method effectively avoids exploring all cells of a full CAD and instead finds the relevant boundary and intersection points by solving lower-dimensional polynomial systems (e.g., by computing resultants or Gr\"obner bases for projection) and deforming possibly singular varieties. 

A third family of relevant algorithms, again singly-exponential, is the family of \emph{roadmap} methods, first introduced by Canny \cite{Canny93:Roadmaps}, whose idea is to build a 1-dimensional ``spine'' inside the set $S$. This is a more sophisticated development of critical point methods, which gives additional information about the connectivity of $S$, recording which sampled points belong to the same connected component. For details, see \cite[Algorithm 16.13]{BasuPollackRoy2006}.

Finally, an algorithm with a different flavor is provided by the \emph{homotopy continuation} methods. These methods solve a polynomial system by continuously deforming a system with known solutions into the target system, and numerically tracking the corresponding solution paths using Newton-type iterations, with complexity governed by the conditioning of the path. A fundamental result shows that for random inputs the expected complexity is polynomial in the input size and the degrees; in particular, see \cite[Theorem 17.1]{BurCuc13:Condition}.

\section{Ham-sandwich algorithm for polytopal measures}\label{sec:HS}
For any $a\in \R^d$, $b\in \R$, denote an associated halfspace by 
\[
H(a\leq b) = \{x\in \R^d \mid \langle a, x\rangle \leq b\},
\]
where $\langle \cdot,\cdot \rangle$ defines the standard scalar product on $\R^d$. $H(a\geq b)$ and $H(a = b)$ are defined analogously. We denote a \emph{cap}, namely the intersection of a set $S$ with the halfspace $H(a\leq b)$, by $S(a\leq b)$. We use analogous notation for intersections with $H(a\ge b)$ and $H(a=b)$.
Given a $d$-dimensional compact set $S\subset \R^d$, for each direction $u\in \Sph^{d-1}$ and value $t\in \R$, consider the associated \emph{cap-volume function}
\[
\capvol_S(u,t) = \vol_d   (S(u\leq t)) = \vol_d(S \cap H(u \leq t)),
\]
where $\vol_d$ is the standard Euclidean volume in $\R^d$.
For a fixed $u$, the map $t\mapsto \capvol_S(u,t)$ is continuous, nondecreasing, and since $S$ is compact there exist $t_0(u), t_1(u) <\infty$ such that
\[
t_0(u) = \inf \{t \in \R \mid \capvol_S(u,t)>0 \}, \quad t_1(u) = \sup \{t \in \R \mid \capvol_S(u,t)<\vol_d (S) \}.
\]
Hence, for every $\alpha\in [0,1]$ and every unit vector $u\in \Sph^{d-1}$ there exist $t_\alpha(u)$ such that $\capvol_S(u,t_\alpha(u))=\alpha\,\vol_d(S)$ for $\alpha\in (0,1)$, with $t_0$ and $t_1$ defined as above. 
We want to record such values in the following set.

\begin{definition}[$\alpha$-quantile locus]\label{def:quantile-locus}
Let $S\subset\R^d$ be a full-dimensional compact set and let $\alpha\in (0,1)$.
The \emph{$\alpha$-quantile locus} of $S$ is the set
\[
V_\alpha(S)
=
\{(u,t)\in \Sph^{d-1}\times\mathbb R :
\capvol_S(u,t)=\alpha\,\vol_d(S)\}.
\]
\end{definition}
In particular, for $\alpha=\tfrac12$ the set $V_{1/2}(S)$ describes those values of $(u,t)$ for which the corresponding halfspace ``halves'' the set $S$.
If $S$ is a set with connected interior, then for every unit vector $u\in \Sph^{d-1}$ and every $\alpha\in(0,1)$ there is a unique value $t_\alpha(u)$ such that $\capvol_S(u,t_\alpha(u))=\alpha\,\vol_d(S)$. Moreover, the map $u\mapsto t_\alpha(u)$ is continuous \cite[Lemma 1]{BaranyHubardJeronimo2008}.
In this case $V_\alpha(S)$ is the graph of this continuous function, and hence has dimension $d-1$.

\begin{remark}\label{rmk:intersection-bodies}
    In the spirit of constructions associated with convex bodies, such as intersection bodies \cite{LUTWAK1988,BBMS22:IntBodies}, one can also visualize $V_\alpha(S)$ through the signed radial map
    \[
        \Phi_{\alpha,S}:\Sph^{d-1}\to \R^d,
        \qquad
        \Phi_{\alpha,S}(u)=t_\alpha(u)u .
    \]
    The image of this map is a geometric object in $\R^d$ associated with the $\alpha$-quantile locus; see \Cref{fig:triangle_curves-gauss} for an example. However,
    this representation should be interpreted with care. 
    Away from the origin, a point $y=\Phi_{\alpha,S}(u)\neq 0$ determines the corresponding hyperplane, since its normal direction is the line spanned by $y$, and $|t_\alpha(u)|=\|y\|$. However, the point $y$ alone does not determine the oriented pair $(u,t_\alpha(u))$ unless the parametrization is retained. At the origin, several directions may be identified, so recovering the corresponding normals requires keeping track of the local branches of $\Phi_{\alpha,S}$, rather than only its image in $\R^d$.
\end{remark}

\begin{example}[$\tfrac{1}{2}$-quantile locus of a triangle]\label{ex:triangle}
    We illustrate the $\alpha$-quantile locus of the triangle
    \[
        P = \conv\{ (-1,-1), (-1,2), (2,-1) \}
    \]
    for $\alpha=\tfrac{1}{2}$.
    A generic hyperplane intersecting $P$ meets exactly two of its edges, giving rise to three combinatorially distinct types of partitions. These are illustrated in \Cref{fig:triangle_curves-planar} by the three representative lines.
    \begin{figure}[ht]
        \centering
        \begin{subfigure}[t]{0.32\textwidth}
            \begin{tikzpicture}[scale=1.5]

\definecolor{tri_color}{named}{MidnightBlue}     

\newcommand{\drawlineABC}[4]{%
  \def\a{#1}
  \def\b{#2}
  \def\c{#3}

  \pgfmathparse{abs(\b) > 0 ? 1 : 0}
  \ifnum\pgfmathresult=1
    \pgfmathsetmacro{\xA}{0}
    \pgfmathsetmacro{\yA}{-\c/\b}
  \else
    \pgfmathsetmacro{\xA}{-\c/\a}
    \pgfmathsetmacro{\yA}{0}
  \fi

  \pgfmathsetmacro{\dx}{\b}
  \pgfmathsetmacro{\dy}{-\a}

  \draw[ultra thick, #4]
    ({\xA - 10*\dx}, {\yA - 10*\dy}) --
    ({\xA + 10*\dx}, {\yA + 10*\dy});
}

\begin{scope}[scale=0.8]
    \clip (-1.5,-1.5) rectangle (2.1,2.1);

    \filldraw[fill=tri_color!50!black, fill opacity=0.2, draw=tri_color!50!black, very thick]
    (-1,-1) -- (-1,2) -- (2,-1) -- cycle;
    \foreach \p in {(-1,-1), (-1,2), (2,-1)}{
      \fill[tri_color!50!black] \p circle (2pt);
    }

    \drawlineABC{1}{1}{1}{RedOrange}
    \drawlineABC{0}{1}{-1}{Dandelion}
    \drawlineABC{1}{0}{-1}{Orchid}

  \node[anchor=north] at (-1,-1) {\small$\ (-1,-1)$};
  \node[anchor=north,  xshift=-1.2em] at (2,-1) {\small$\ (2,-2)$};
  \node[anchor=west, yshift=-0.2em] at (-1,2) {\small$\ (-1,2)\ $};
\end{scope}
\end{tikzpicture}
            \caption{The triangle $P$ with three hyperplane cuts.}
            \label{fig:triangle_curves-planar}
        \end{subfigure}
        \begin{subfigure}[t]{0.32\textwidth}
            \includegraphics[height=4.5cm]{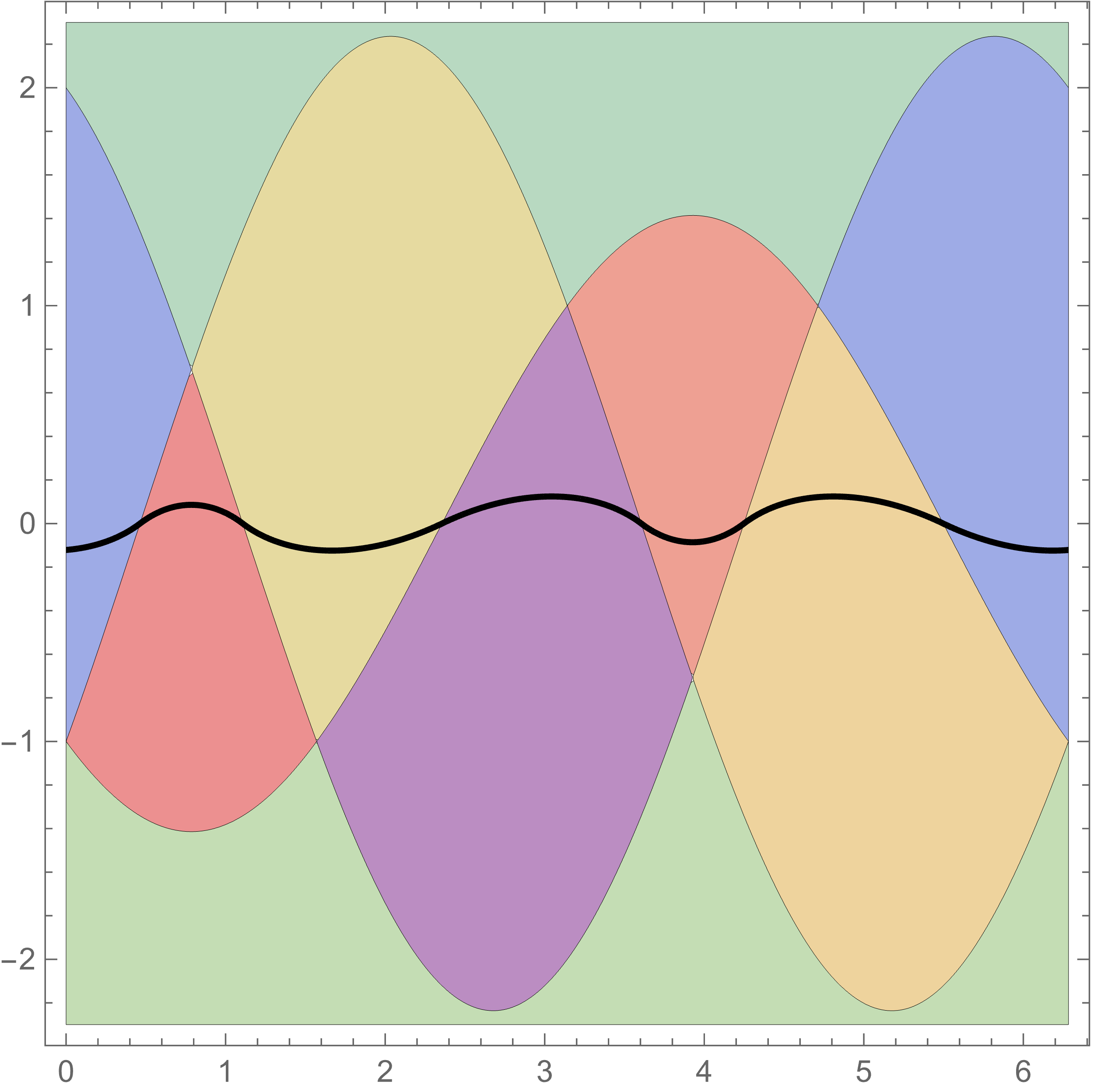}
            \caption{The $\tfrac12$-quantile locus $V_{\tfrac12}(P)$.}
            \label{fig:triangle_curves-chambers}
        \end{subfigure}
        \begin{subfigure}[t]{0.32\textwidth}
            \includegraphics[height=4.5cm]{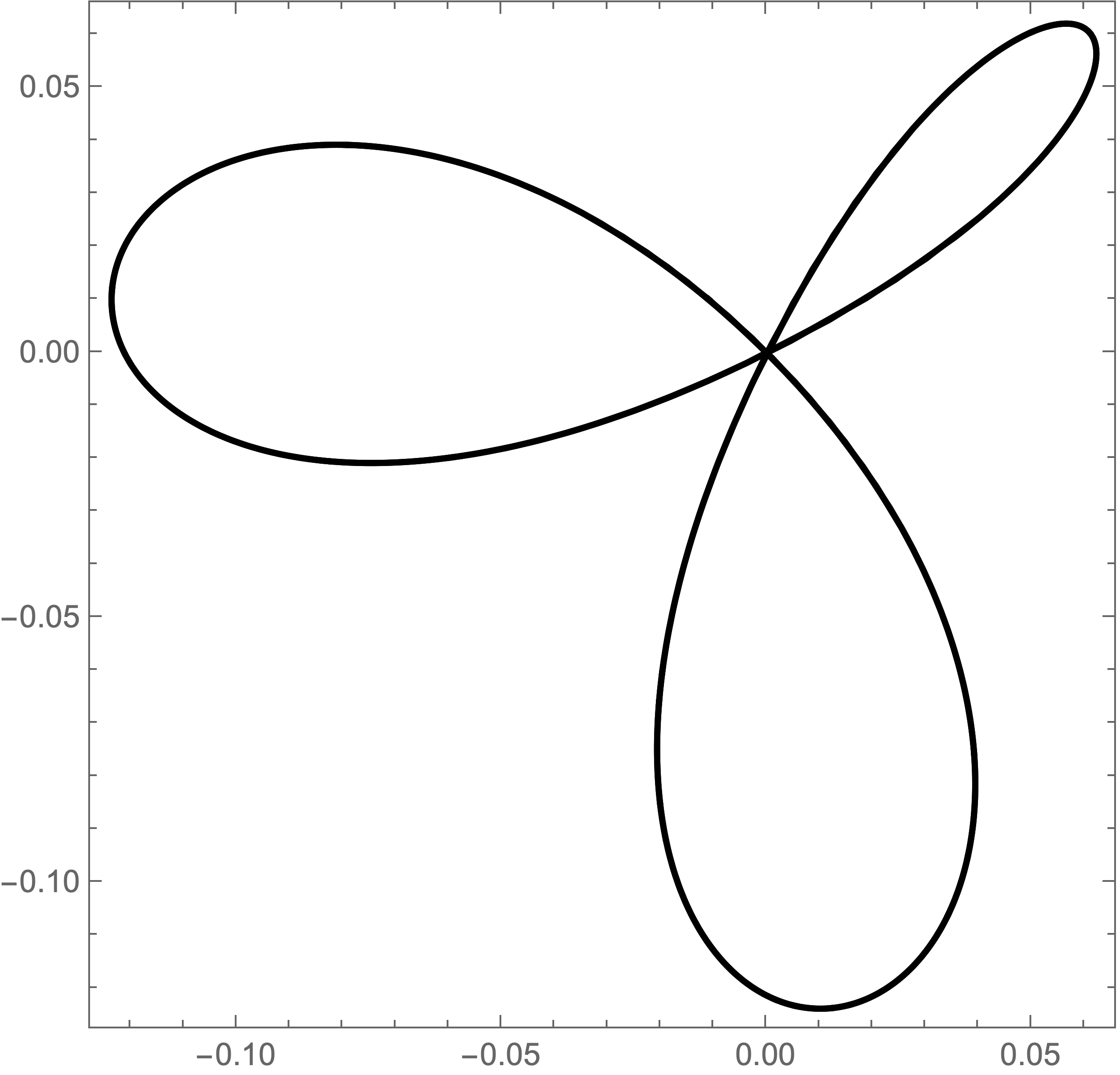}
            \caption{The curve $\Phi_{\tfrac12,P}(\Sph^1)$.}
            \label{fig:triangle_curves-gauss}
        \end{subfigure}
        \caption{The objects described in \Cref{ex:triangle}. The locus $V_{\frac{1}{2}}(P)$ is the black curve shown in \Cref{fig:triangle_curves-chambers},  parametrized by $\{(\cos \theta, \sin \theta, t)\}$ within the $(\theta,t)$-plane. The colorful regions in \Cref{fig:triangle_curves-chambers} are the regions of the parameter space on which the rational expression for $\capvol_P$ is fixed.}
        \label{fig:triangle_curves}
    \end{figure}
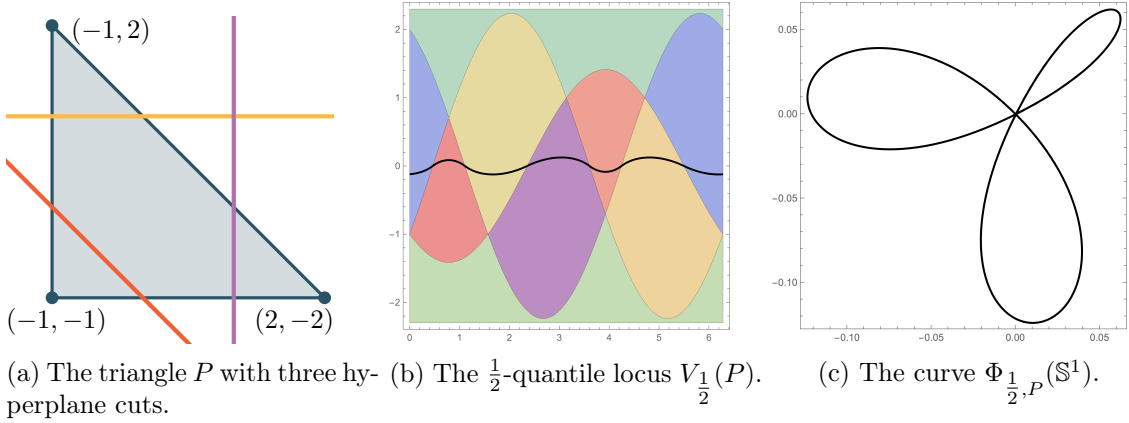
    Fix an edge of $P$ with endpoints $a$ and $b$. The intersection point of this edge with a hyperplane $H(u=t)$ can be expressed as a rational function of variables $u,t$ and coefficients depending on $a$ and $b$. Consequently, $\capvol_P(u,t)$ is itself a rational function when restricted to hyperplanes intersecting a fixed pair of edges. As we shall see later, this generalizes to polytopes and hyperplanes in any dimension.
    
    For example, let $u\in \Sph^1$ and suppose that
    \begin{equation}\label{eq:triangle_assumptions}
          t + u_1 + u_2 \geq 0, \quad 
           t - 2u_1 + u_2 \leq 0, \quad 
           t + u_1 - 2u_2 \leq 0.
    \end{equation}
    Under these conditions, the cap $P(u \leq t)$ consists of three vertices: the point $(-1,-1)$, together with one point on the edge $\conv\{(-1,-1),(2,-1)\}$ and one point on the edge $\conv\{(-1,-1),(-1,2)\}$ (corresponding to the red line in \Cref{fig:triangle_curves-planar}). In this case, the cap volume function can be written as
    \[
        \capvol_P(u_1,u_2,t) = \frac{{\left(t + u_{1} + u_{2}\right)}^{2}}{2 \, u_{1} u_{2}} \ .
    \]
    Solving the quadratic equation $\capvol_P(u_1,u_2,t) = \tfrac{1}{2} \vol_2(P)$ subject to the above constraints yields a unique solution
    \begin{equation*}\label{eq:triangle}
        t_{\frac12}(u) = - u_{1} - u_{2} + 3\sqrt{ \frac{u_{1} u_{2}}{2}  } \ .
    \end{equation*}
    The set of points $(u,t) \in \mathbb{R}^3$ with $u \in \Sph^1$ and $t = t_{\frac12}(u)$ satisfying \eqref{eq:triangle_assumptions} forms a portion of the $\tfrac12$-quantile locus. This portion corresponds to the part of the black curve shown in \Cref{fig:triangle_curves-chambers} that lies within the two red regions.
    Different regions in this figure correspond to combinatorially distinct intersections of $P$ with a hyperplane. We parametrize $\Sph^1$ as $(\cos(\theta),\sin(\theta))$, and identify $(\theta,t)$ with $(\pi+\theta,-t)$, since for $\alpha=\tfrac{1}{2}$ the upper and lower caps have equal volume and hence define the same hyperplane. Thus, regions in \Cref{fig:triangle_curves-chambers} with similar color correspond to the same unoriented hyperplanes. Accordingly, the hyperplanes in the left panel are colored according to the regions of the $(\theta,t)$-parameter space shown in the central panel. \Cref{fig:triangle_curves-gauss} shows the image of the map $\Phi_{\tfrac{1}{2},P}$ from \Cref{rmk:intersection-bodies}.
\end{example}

In the rest of the paper we will focus on the behavior of the $\alpha$-quantile locus of a convex polytope, or a union of such, in terms of semialgebraic geometry. 
The following lemma is a special version of Theorems 1.1 and 1.3 in \cite{BDC2025bestslices}, where we proved that given a polytope $P \subset \R^d$, there exists a decomposition of the space of all affine hyperplanes in $\R^d$ into finitely many cells, called \emph{slicing chambers}, with the properties that:
\begin{enumerate}
\item Two affine hyperplanes $H_1,H_2$ belong to the same slicing chamber if and only if $H_1,H_2$ intersect the same set of faces of $P$. In particular, $P\cap H_1$ is combinatorially equivalent to $P \cap H_2$ and they admit a combinatorially identical triangulation. \label{item:comb-info}
\item For fixed dimension $d$, the number of slicing chambers is bounded by a polynomial in the number of vertices of $P$.\label{item:polynomial-size}
\item Restricted to a fixed slicing chamber, the integral $\int_{P\cap H} g(x) dx$ of any given polynomial $g: \R^d \to \R$ is a rational function, whose variables are the coefficients of the defining equation of the hyperplane $H$. In particular, this holds for the volume of $P\cap H_+$ for the oriented halfspace defined by $H$. This rational function depends only on the combinatorial information described in \labelcref{item:comb-info}. \label{item:integral-rational}
\end{enumerate}

Following exactly the same methods we show that the cap-volume function of a union of polytopes is semialgebraic. When the dimension of the space is fixed, it is not restrictive to assume that the polytopes have disjoint interior, since one can always, in polynomial time, rewrite a union of (possibly intersecting) polytopes as a union of polytopes with disjoint interior. We comment on this well-known procedure in the following remark.
\begin{remark}\label{rmk:disjoint-interior}
Consider the polytopes $Q_1,\ldots,Q_m$, let $P =\bigcup_{i=1}^m Q_i$ and suppose first that the $Q_i$ are given by finite systems of linear inequalities. In fixed dimension, such an inequality description can be computed from the vertex description of each $Q_i$ in polynomial time. Let $\mathcal H$ be the collection of all affine hyperplanes supporting inequalities appearing in these descriptions. The arrangement induced by $\mathcal H$ subdivides $\mathbb R^d$ into polyhedral cells. Since each cell has a fixed sign pattern with respect to all inequalities defining the $Q_i$, the closure of every full-dimensional cell is either contained in a given $Q_i$ or is disjoint from its interior. Hence, by taking the closures of those full-dimensional cells which are contained in at least one $Q_i$, we obtain polytopes $R_1,\ldots,R_s$ with pairwise disjoint interiors such that $P=\bigcup_{j=1}^s R_j$.
When the dimension $d$ is fixed, the number of cells in an arrangement of $N$ hyperplanes in $\mathbb R^d$ is $O(N^d)$, and such an arrangement can be constructed in polynomial time in the input size. Thus, this preprocessing step has polynomial complexity in fixed dimension.
\end{remark}

\begin{remark}\label{rmk:unions}
    In what follows, we consider unions of polytopes. As input, we will assume that the union $P = Q_1 \cup \dots \cup Q_m$ is presented as some list $(Q_1,\dots,Q_m)$ of convex polytopes. The \emph{set of vertices} $\verts(P)$ or \emph{set of edges} of $P$ is the disjoint union of the sets of vertices or edges of $Q_1,\dots,Q_m$. In particular, even if $Q_i,Q_j$ have overlapping or nested edges, we still treat these as distinct objects. The \emph{total number} of vertices or edges of $P$ is the sum of vertices or edges of the subpolytopes $Q_1,\dots,Q_m$. Note that with this convention, these numbers depend on the non-unique decomposition into subpolytopes $Q_1,\dots,Q_m$. If $\mathcal P = (P_1,\dots,P_k)$ is a tuple of unions of polytopes, then the \emph{set of vertices} $\verts(\mathcal P)$ is the union of all vertices of all subpolytopes of $P_1,\dots,P_k$, after fixing a decomposition for every single one of them.
\end{remark}

\begin{lemma}[Semialgebraicity of the cap-volume function]\label{lem:cap-volume-semialg}
Let $P= Q_1 \cup \ldots \cup Q_m\subset\R^d$ be a union of finitely many full-dimensional polytopes. 
There exists a central hyperplane arrangement $\mathcal{H}_P\subset \R^{d+1}$ whose restriction to $\Sph^{d-1} \times \R$ partitions $\Sph^{d-1}\times\R$ into finitely many full-dimensional cells. For each such cell $\Omega$ we have 
\[
{\capvol_P}{\big\vert_\Omega}(u,t) = \frac{p_\Omega(u,t)}{q_\Omega(u,t)},
\]
for some polynomials $p_\Omega, q_\Omega \in \R[x_1,\ldots,x_{d+1}]$ such that $q_\Omega$ does not vanish on $\Omega$. In other words, the function $\capvol_P$ is a piecewise rational function. 
In fixed dimension, the number of pieces is polynomial in the total number of vertices of $P$. If $Q_1,\dots,Q_m$ have pairwise disjoint interiors, then the number of pieces is at most $O(|\verts(P)|^d)$ and the degrees of $p_\Omega(u,t)$ and $q_\Omega(u,t)$ are at most the total number of edges of $P$.
\end{lemma}

This result is analogous to \cite[Proposition 5.5]{BBMS22:IntBodies} and \cite[Proposition 3.3]{BDC2025bestslices}, but we write the proof here for completeness. 
\begin{proof}
    By \Cref{rmk:disjoint-interior}, we can assume, without loss of generality, that the convex polytopes $Q_i$ have disjoint interior.  Let $\verts(P)$ be the set of vertices of $P$, namely the set of vertices of all convex subpolytopes $Q_i$ of $P$. Define the slicing arrangement, i.e., the hyperplane arrangement
    \[
    \mathcal{H}_P = \{ (v,-1)^\perp \mid v \in \verts(P)\} \subset \R^{d}\times\R,
    \]
    with coordinates $(u,t)$. Denote by $\Omega$ a full-dimensional connected component of $(\Sph^{d-1}\times\R)\setminus \mathcal{H}_P$; we call it a chamber. A point $(u,t)\in \Omega$ corresponds to a halfspace $H(u\leq t)$. By construction, the two sets $P(u_1\leq t_1)$ and $P(u_2\leq t_2)$ with $(u_1,t_1), (u_2,t_2)\in \Omega$ are combinatorially equivalent, since the hyperplanes $H(u_i = t_i)$ induce the same partition of the points in $\verts(P)$.
    Since the convex polytopes $Q_i$ have disjoint interior, we can prove our statement for a single one of them.
    We can then conclude by the fact that rational functions are closed under addition.

    Fix the index $i \in [m]$. The vertices of $Q_i(u\leq t)$ that are not vertices of $Q_i$ itself arise as the intersection of an edge $[v,w]$ with the bounding hyperplane $H(u=t)$, and their coordinates are rational functions in $(u,t)$:
    \begin{equation}\label{eq:verts_cap}
    [v,w] \cap H(u=t) = 
    v +\frac{t-\langle u,v\rangle}{\langle u,w-v\rangle}(w-v),
    \end{equation}
    where the denominator is nonzero precisely because the edge crosses the hyperplane transversely.
    Choose a fixed triangulation scheme from the vertex set of each cap $Q_i(u \leq t)$ (for instance, a pulling triangulation with a fixed vertex order). 
    Since the combinatorial type does not change while $(u,t)\in \Omega$, the same abstract triangulation applies to all $Q_i(u \leq t)$. 
    The cap-volume is then a finite sum of volumes of the simplices in the triangulation, each given by the determinant of a matrix whose columns are the vertices of the simplices, namely those in \eqref{eq:verts_cap} and the vertices of $Q_i$ on the considered side of the hyperplane. Each such determinant is a rational function in $(u,t)$ defined over $\Omega$. Thus, $\capvol_{Q_i}$ is a piecewise rational function, where the pieces are chambers $\Omega$. Therefore, after taking the union, $\capvol_P$ is also a piecewise rational function.

    Let $E_\Omega$ denote the set of edges of $P$ intersected by a hyperplane with parameters in $\Omega$, and let $n_\Omega$ denote its cardinality. 
    After clearing denominators, the cap-volume function can be written as
    \[
    \capvol_P(u,t)\big|_\Omega=\frac{p_\Omega(u,t)}{q_\Omega(u)},
    \]
    where
    \[
    q_\Omega(u)=\prod_{[v,w]\in E_\Omega}\langle u,w-v\rangle .
    \]
    Thus $\deg q_\Omega\le n_\Omega$. Since each determinant appearing in a fixed triangulation is multilinear in the vertices, after clearing this common denominator, the numerator has degree at most $n_\Omega$. See \cite[Proposition 3.3]{BDC2025bestslices} for details.
    In particular, $n_\Omega$ is at most the total number of edges of $P$, and for fixed dimension this is polynomial in the number of vertices of $P$. 

    Lower-dimensional cells of $\mathcal H_P$ correspond to nongeneric slices passing through vertices or higher-dimensional faces. As $\capvol$ is a continuous function, the rational function on a low-dimensional cell of the arrangement is a specialization of the rational functions on adjacent full-dimensional chambers.
    
    To conclude the bound on the number of pieces, note that $\mathcal H_P$ is a central arrangement with $|\verts(P)|$ hyperplanes, hence the number of full-dimensional chambers is at most $O(|\verts(P)|^d)$ for fixed $d$ (cf. \cite{Zaslavsky}). Intersecting this arrangement with
    $\Sph^{d-1}\times\mathbb R$ gives at most the same order of cells.
\end{proof}

A piecewise rational function is, in particular, semialgebraic. 
As a direct consequence, the $\alpha$-quantile locus of a union of polytopes is a semialgebraic subset of $\R^{d+1}$. 
Moreover, recall from the beginning of this section that for convex $d$-dimensional sets $P$, the locus $V_\alpha(P)$ is the graph of a continuous function. This yields the following statement.
\begin{corollary}\label{prop:quantile-locus-semi-alg}
Let $P= Q_1 \cup \ldots \cup Q_m\subset\R^d$ be a union of finitely many full-dimensional polytopes with disjoint interior and let $\alpha\in (0,1)$. Then, the $\alpha$-quantile locus $V_\alpha(P)$ is a semialgebraic set. 
If, in addition, $P$ is convex, then $V_\alpha(P)$ is the graph of a continuous semialgebraic function, and hence has dimension $d-1$.
\end{corollary}

In principle, we can repeat an analogous construction for the extreme values $\alpha = 0,1$. However, in this case, the space of hyperplanes that do not intersect the compact set $P$ is full-dimensional, and hence $\dim(V_\alpha(P)) = d$.

\begin{definition}[$\alpha$-cuts]
    Consider $\alpha = (\alpha_1,\ldots,\alpha_k) \in (0,1)^k$. We say that a hyperplane $H(u = t)$ is an $\alpha$-cut of the sets $S_1,\ldots,S_k \subset \R^d$ if $\capvol_{S_i}(u,t) = \alpha_i \vol_d (S_i)$ for all $i=1,\dots,k$. We denote the set of all $\alpha$-cuts by 
    \[
    \HS_\alpha(S_1,\ldots,S_k) = \{(u,t) \in \Sph^{d-1} \times \R \mid H(u=t) \text{ is an $\alpha$-cut}\}.
    \]
\end{definition}

In the polytopal case, we can construct the $\alpha$-quantile loci of multiple unions of polytopes to obtain their $\alpha$-cuts. 
Indeed, let $\mathcal{P} = (P_1,\ldots,P_k)$ be a tuple of unions of polytopes, where each $P_i \subset \R^d$ is itself the union of finitely many full-dimensional polytopes. For any $\alpha\in (0,1)^k$, the set of $\alpha$-cuts of $\mathcal P$ equals
\[
\HS_\alpha(\mathcal{P}) = \bigcap_{i=1}^k V_{\alpha_i}(P_i).
\]
If $k\leq d$ and $\alpha_i = \frac{1}{2}$ for all $i \in [k]$, then $\HS_\alpha(\mathcal{P}) \neq \emptyset$ according to the ham-sandwich theorem.
We note the following feature of the dimension of the set of $\alpha$-cuts.

\begin{proposition}\label{prop:space-of-hs-solutions}
    Let $\mathcal{P} = (P_1,\ldots,P_k)$ where each $P_i \subset \R^d$ is the union of finitely many full-dimensional polytopes, and $\alpha = (\alpha_1,\ldots,\alpha_k)\in (0,1)^k$. The set $\HS_\alpha(\mathcal{P})$ of its $\alpha$-cuts is semialgebraic. Moreover, allowing the number $k$ of measures to vary, every integer between $-1$ and $d$ can occur as the dimension of such a set $\HS_\alpha(\mathcal P)$.
\end{proposition}
\begin{proof}
    Since $V_{\alpha_i}(P_i)$ is semialgebraic by \Cref{prop:quantile-locus-semi-alg}, and since this property is preserved by finite intersections, $\HS_\alpha(\mathcal{P})$ is also semialgebraic for any $\alpha\in (0,1)^k$. 
    There exist instances, in any dimension $d$, of tuples with empty $\alpha$-cuts. For this, consider any tuple $\mathcal{P}$ with symmetric convex bodies $P_1=-P_1$ and $P_2 = \lambda P_1$ for some $\lambda\neq 1$, both containing the origin in their interior. Then, if $\alpha_1 = \frac{1}{2}$ and $\alpha_2 \neq \frac{1}{2}$, there are no $\alpha$-cuts. Indeed, the hyperplanes that cut $P_1$ in half are all and only those that contain the origin. Since this is also the center of symmetry of $P_2$, it is impossible to obtain a simultaneous cut of the two bodies with different proportions. This implies 
    \[
    \HS_\alpha(\mathcal{P}) = \emptyset, \qquad\dim \HS_\alpha(\mathcal{P}) = -1.
    \]
    Next, we construct $\mathcal P$ and $\alpha \in (0,1)^k$ with $\dim(\HS_\alpha(\mathcal P)) = m$ for any $m = 0,\dots,d-1$.
    For any given $m$, let $l = d-m$ and $k \geq l$. For 
    $i=1,\dots,l$, choose $P_i = [-1,1]^d + v_i$ where $v_1,\dots,v_{l}$ are linearly independent vectors. For $i = l+1,\dots,k$, choose $P_i = \lambda_i[-1,1]^d + v_l$ with pairwise distinct $\lambda_i$. For all $i=1,\dots,k$ let $\alpha_i = \frac12$. Then, the $\frac{1}{2}$-quantile locus of $P_i$ consists of all pairs $(u,t)$ such that $H(u=t)$ contains the point $v_i$, namely
    \[
        V_{\frac12}(P_i) = \{(u,t) \in \Sph^{d-1} \times \R \mid \langle u, v_i \rangle = t \} = \{ (u,\langle u,v_i \rangle) \in \Sph^{d-1} \times \R  \} \ ,
    \]
    which is a $(d-1)$-dimensional sphere. 
    The equations $\langle u,v_i-v_1\rangle=0$, for $i=2,\ldots,l$, cut $\Sph^{d-1}$ by $l-1$ independent linear equations. Hence the solution set is a sphere of dimension
    \[
    (d-1)-(l-1)=d-l=m.
    \]
    Finally, let $P_i = (\frac{1}{2}[-1,1]^d + (i e_1 + e_d) )\cup (\frac{1}{2}[-1,1]^d + (i e_1 - e_d) )$ be the union of two cubes, appropriately shifted for all $i=1,\ldots,k$. Then, there is a $d$-dimensional set of hyperplanes around $(u,t) = (e_d,0)$ that belong to $\HS_\alpha(\mathcal{P})$ for $\alpha_1 = \ldots = \alpha_k =1/2$.
\end{proof}

In the remainder of this section, we focus on the \AlphaHS{} algorithm, namely on \Cref{thm:generalized_HS_intro}. As input, it takes rational polytopal measures $P_1,\dots,P_k$, and computes the set of all their $\alpha$-cuts, for $\alpha = (\alpha_1,\dots,\alpha_k) \in (0,1)^k\cap \Q^k.$
We sketch the algorithm in \Cref{alg:alpha-HS}.

\begin{algorithm}[ht]
\caption{\AlphaHS{} Algorithm}
\label{alg:alpha-HS}
\begin{algorithmic}[1]
\Require A tuple $\mathcal P=(P_1,\ldots,P_k)$ of rational polytopal measures in $\mathbb R^d$, where each $P_i$ is given as a finite union of full-dimensional rational polytopes, and target values $\alpha=(\alpha_1,\ldots,\alpha_k)\in(0,1)^k \cap \Q^k$.
\Ensure A quantifier-free semialgebraic description of $\HS_\alpha(\mathcal P)$, and sample points in its connected components.

\For{$i=1,\ldots,k$}
    \State Construct the slicing arrangement
        $\mathcal H_{P_i}
        =
        \{(v,-1)^\perp \mid v\in \verts(P_i)\}
        \subset \mathbb R^{d+1}$
    in the parameter space with coordinates $(u,t)$.
    \For{each cell $\Omega$ of $\mathcal H_{P_i}$}
        \State Compute the rational expression $\capvol_{P_i}(u,t)$ for $(u,t) \in \Omega$
    \EndFor
    \State Define the piecewise rational set $V_{\alpha_i}(P_i)$.
\EndFor

\State Apply a semialgebraic-set algorithm to $\HS_\alpha (\mathcal{P}) = \cap V_{\alpha_i}(P_i)$ to decide emptiness, dimension, and sample one point in each connected component.
\end{algorithmic}
\end{algorithm}

We now show the correctness of the algorithm, and that the complexity of this algorithm is polynomial in the total number of input vertices.

\begin{proof}[Proof of \Cref{thm:generalized_HS_intro}]
    Let $n$ be the total number of vertices of all input subpolytopes of $P_1,\dots,P_k$ after the disjoint-interior refinement (cf. \Cref{rmk:disjoint-interior,rmk:unions}). The common slicing arrangement 
    $\mathcal H_{\mathcal P}
        =
        \{(v,-1)^\perp \mid v\in \verts(\mathcal P)\}$
         is a central arrangement of at most $n$ hyperplanes in $\mathbb R^{d+1}$.
        Although \Cref{alg:alpha-HS} constructs the arrangements $\mathcal H_{P_i}$ separately, for the complexity analysis we may equivalently refine them to the common arrangement $\mathcal H_{\mathcal P}$. On this refinement, for every $i \in [k]$, the function $\capvol_{P_i}$ has a fixed rational function expression on each cell.
         The number of full-dimensional cells of $\mathcal H_{\mathcal P}$ is at most $O(n^d)$ for fixed $d$ (cf. \Cref{lem:cap-volume-semialg} for the case where $P$ is the union of all subpolytopes).
    On each cell $\Omega$, each cap-volume function $\capvol_{P_i}(u,t)$ is represented by a rational function
    \[
    \capvol_{P_i}(u,t)\big|_\Omega
    =
    \frac{p_{i,\Omega}(u,t)}{q_{i,\Omega}(u)}
    \]
    with, as shown in \Cref{lem:cap-volume-semialg},
    $\deg p_{i,\Omega},\deg q_{i,\Omega}\le e_i\le O(n^2)$,
    where $e_i$ is the total number of edges of the subpolytopes supporting $P_i$. Thus, the equation
    \[
    \capvol_{P_i}(u,t)=\alpha_i\vol_d(P_i)
    \]
    is, when restricted to $\Omega$, equivalent to a polynomial equation of degree at most $O(n^2)$. The inequalities defining $\Omega$ are linear, and the sphere condition $\|u\|^2=1$ has degree $2$.

    Consequently, $\HS_\alpha(\mathcal P)$ is described by a quantifier-free formula (cf. \Cref{prop:space-of-hs-solutions}) in $d+1$ variables using $s=O(n^d(n+k))$ polynomials, all of degree at most $D = O(n^2)$. Since $d$ is fixed, standard algorithms for semialgebraic sets have complexity $(sD)^{O(d)}$. Therefore, they run in polynomial time in the input parameters $n$ and $k$.
\end{proof}

\begin{proof}[Proof of \Cref{thm:HS_1/2_intro}]
    This is \Cref{thm:generalized_HS_intro} for $k = d$ and $\alpha_1 = \dots = \alpha_d =\frac{1}{2}$.
\end{proof}

\begin{example}[The set of $\alpha$-cuts of two measures in dimension $2$]\label{ex:ham-sandwich_2d}
Let $\mathcal{P}=\{P_1,P_2\}$, where $P_1$ is the triangle from \Cref{ex:triangle}, and 
\[
P_2 = \conv\{(-1,0), (2,0), (1,-3)\} \ \cup \ \conv\{(1,0), (1,1), (0,1), (0,0)\} .
\]
We run our algorithm with $\alpha = (\tfrac{1}{2},\tfrac{1}{4})$ and construct the associated piecewise rational functions $\capvol_{P_i}$. The regions on which these functions are fixed rational functions are depicted in different colors in \Cref{fig:triangle_curves-chambers} (for $P_1$) and \Cref{fig:planar_hs-chambers} (for $P_2$).
These regions are not polyhedral, as we visualize them in $(\theta,t)$-coordinates rather than in the three-dimensional embedding $(\cos\theta,\sin\theta,t)$. In both \Cref{fig:triangle_curves-chambers,fig:planar_hs-chambers}, the black curves represent the $\alpha_i$-quantile loci for $\alpha_1=\tfrac{1}{2}$ and $\alpha_2=\tfrac{1}{4}$, respectively. To determine the lines that simultaneously cut both sets in the prescribed proportions, we intersect these two curves (see \Cref{fig:planar_hs-intersection}). This yields two red intersection points, corresponding to the two red lines in \Cref{fig:planar_hs-polygons}.
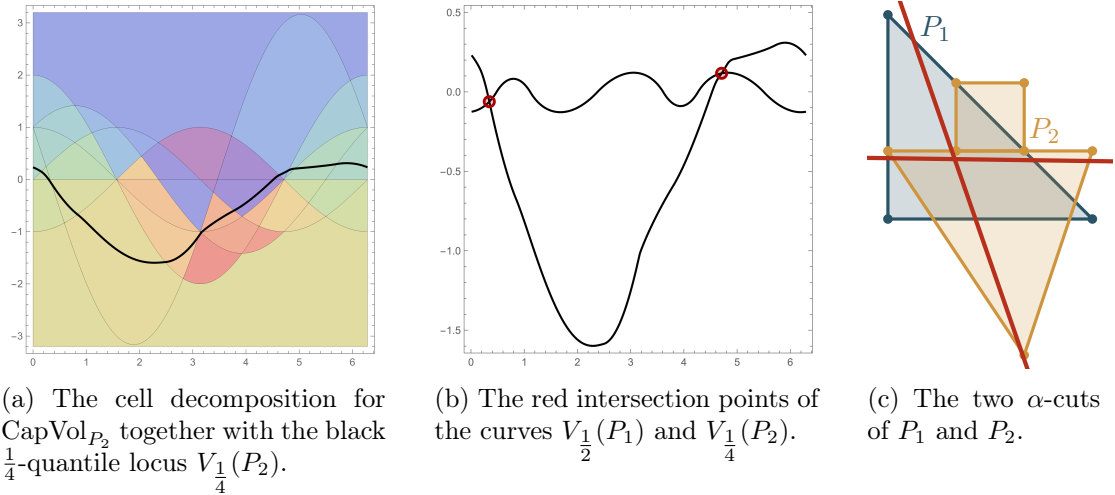
\begin{figure}[hb]
    \centering
    \begin{subfigure}[t]{0.33\textwidth}
        \centering
        \includegraphics[height=4.8cm]{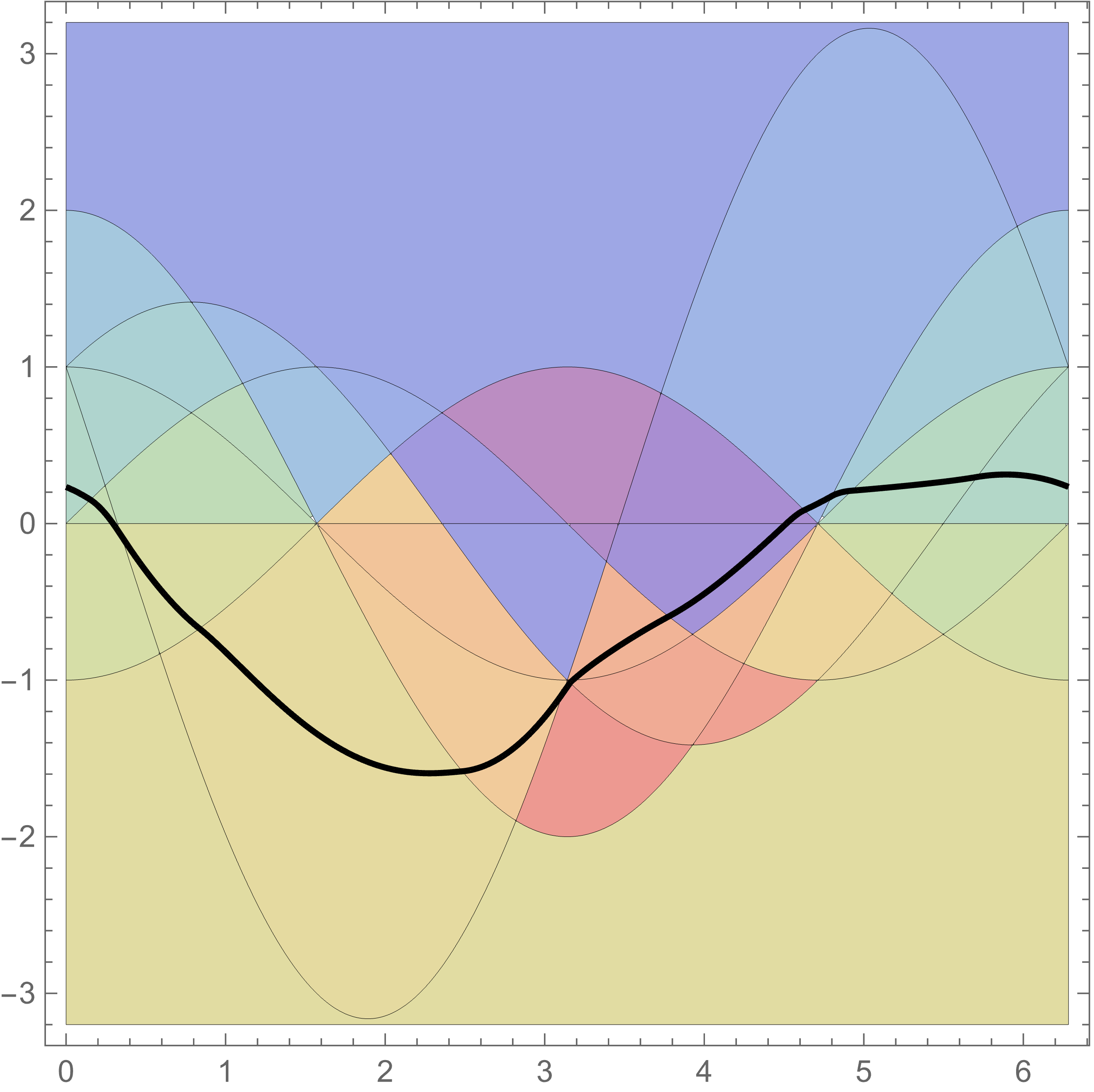}
        \caption{The cell decomposition for $\capvol_{P_2}$ together with the black $\tfrac{1}{4}$-quantile locus $V_{\tfrac{1}{4}}(P_2)$.}
        \label{fig:planar_hs-chambers}
    \end{subfigure}
    \hspace{1em}
    \begin{subfigure}[t]{0.33\textwidth}
        \centering
        \includegraphics[height=4.8cm]{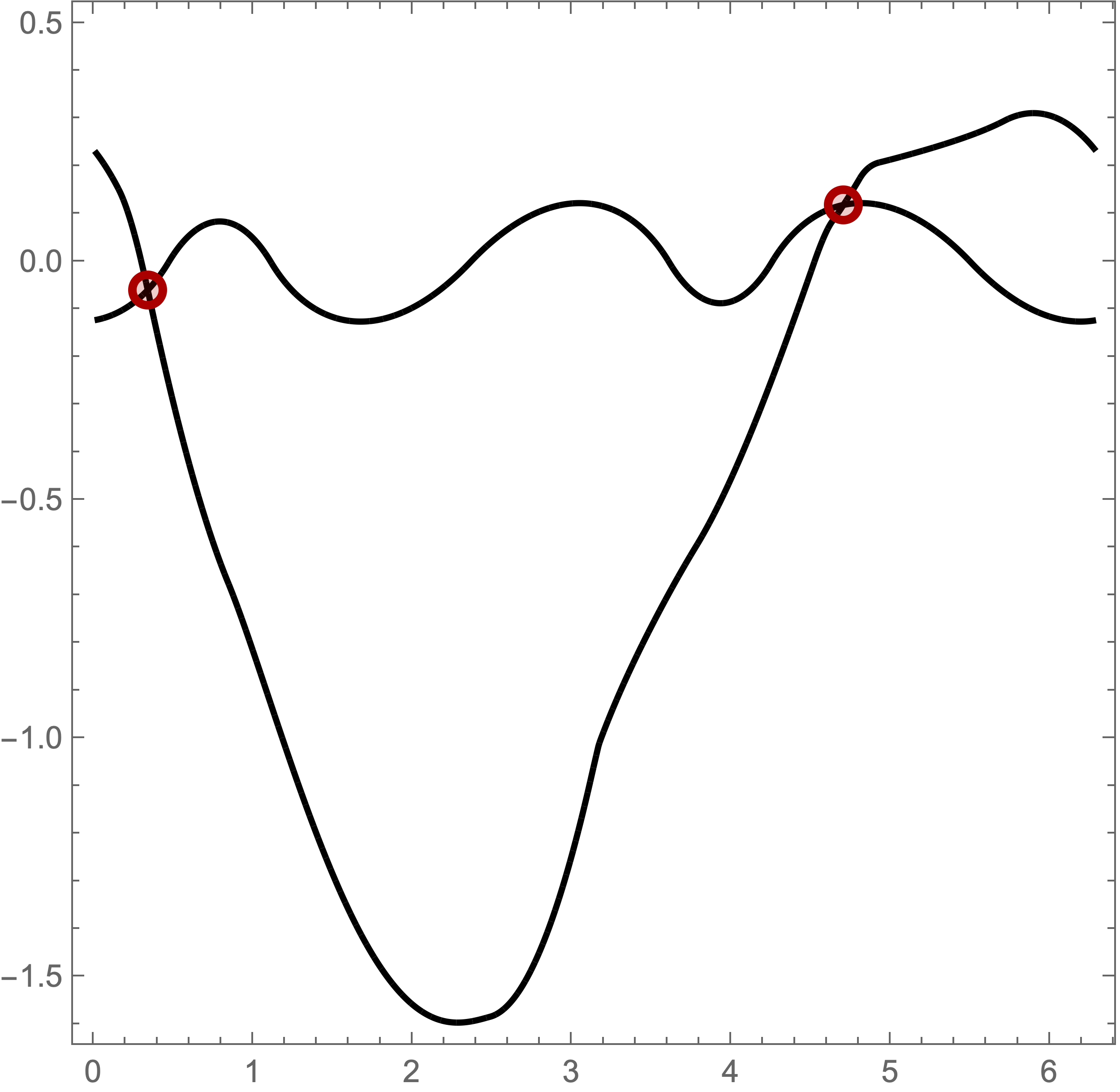}
        \caption{The red intersection points of the curves $V_{\tfrac{1}{2}}(P_1)$ and $V_{\tfrac{1}{4}}(P_2)$.}
        \label{fig:planar_hs-intersection}
    \end{subfigure}
    \hspace{1em}
    \begin{subfigure}[t]{0.2\textwidth}
        \centering
        \raisebox{0pt}[\height][0pt]{%
        \begin{tikzpicture}[scale=0.9]

\definecolor{tri_color}{named}{MidnightBlue}   
\definecolor{union_color}{named}{YellowOrange}   
\definecolor{hs1}{named}{BrickRed}   
\definecolor{hs2}{named}{BrickRed}   

\newcommand{\drawlineABC}[4]{%
  \def\a{#1}
  \def\b{#2}
  \def\c{#3}

  \pgfmathparse{abs(\b) > 0 ? 1 : 0}
  \ifnum\pgfmathresult=1
    \pgfmathsetmacro{\xA}{0}
    \pgfmathsetmacro{\yA}{-\c/\b}
  \else
    \pgfmathsetmacro{\xA}{-\c/\a}
    \pgfmathsetmacro{\yA}{0}
  \fi

  \pgfmathsetmacro{\dx}{\b}
  \pgfmathsetmacro{\dy}{-\a}

  \draw[ultra thick, #4]
    ({\xA - 10*\dx}, {\yA - 10*\dy}) --
    ({\xA + 10*\dx}, {\yA + 10*\dy});
}

\begin{scope}[scale=1]
    \clip (-1.3,-3.2) rectangle (2.3,2.2);

    \filldraw[fill=tri_color!50!black, fill opacity=0.2, draw=tri_color!50!black, very thick]
    (-1,-1) -- (-1,2) -- (2,-1) -- cycle;
    \foreach \p in {(-1,-1), (-1,2), (2,-1)}{
      \fill[tri_color!50!black] \p circle (2pt);
    }

    \filldraw[fill=union_color!80!black, fill opacity=0.2, draw=union_color!80!black, very thick]
    (-1,0) -- (2,0) -- (1,-3) -- cycle;
    \filldraw[fill=union_color!80!black, fill opacity=0.2, draw=union_color!80!black, very thick]
    (1,0) -- (1,1) -- (0,1) -- (0,0) -- cycle;
    \foreach \p in {(-1,0), (2,0), (1,-3),(1,0), (1,1), (0,1), (0,0)}{
      \fill[union_color!80!black] \p circle (2pt);
    }

    \drawlineABC{0.9449360820143504}{0.3272549478730134}{0.0580343067614477}{hs1}
    \drawlineABC{-0.014162128911249652}{-0.9998997120235115}{-0.12039553098094291}{hs2}
  
  \node[] at (-0.3,1.8) {\textcolor{tri_color!50!black}{$P_1$}}; 
  \node[] at (1.3,0.3) {\textcolor{union_color!80!black}{$P_2$}}; 
\end{scope}
\end{tikzpicture}%
        }
        \caption{The two $\alpha$-cuts of $P_1$ and $P_2$.}
        \label{fig:planar_hs-polygons}
    \end{subfigure}
    \caption{The $\alpha$-cuts of $P_1$ and $P_2$ from \Cref{ex:ham-sandwich_2d}.}
    \label{fig:planar_hs}
\end{figure}

We can also derive an explicit algebraic description of these intersection points. For instance, the leftmost intersection point is the unique feasible solution of the following polynomial system of equations and inequalities:
\begin{gather*}
    -u_1 u_2 + 8 u_1 t -4 u_2 t + u_1^2 - 2u_2^2 - 2t^2=0, \\
    33 u_1 u_2 + 24 u_1 t - 10 u_1^2 + 12 t^2=0,\\
    u_1^2+u_2^2=1,\\
    -u_2 + t \leq 0,\quad 
    t \leq 0,\quad 
    -u_1 + 3 u_2 + t \leq 0,\quad 
    -u_1 - u_2 - t \leq 0,\quad 
    -u_1 + 2 u_2 - t \leq 0.
\end{gather*}
Up to four-digit precision, the unique solution is $(u,t) = (0.9449, 0.3272, -0.0580)$. This solution can be computed, for example, using the \texttt{CylindricalDecomposition} command in \texttt{Mathematica}, which returns the result in $0.007$ seconds.
\end{example}

\begin{example}[The set of $\alpha$-cuts of three measures in dimension $3$]\label{ex:ham-sandwich_3d}
    Consider two slices of bread, obtained as the convex hulls of the two sets of vertices 
    \begin{gather*}
    \left\{
    (0, 0, -3),\, (8, 0, -3),\, (0, 5, -3),\, (0, 0, -1),\, (8, 0, -1),\, (0, 5, -1)
    \right\}, \\
    \left\{
    (0, 0, 3),\, (8, 0, 3),\, (0, 5, 3),\, (0, 0, 1),\, (8, 0, 1),\, (0, 5, 1)
    \right\},
    \end{gather*}
    a slice of ham, the convex hull of 
    \[
    \left\{
    (1, -1, 0),\, (1, -1, 1),\, (7, 1, 0),\, (6, 1, 1),\, (1, 5, 0),\, (1, 4, 1),\, (-1, 4, 0),\, (-1, 4, 1)
    \right\},
    \]
    and a piece of cheese, namely the polytope with vertices
    \[
    \left\{
    (0, 0, 0),\, (0, 0, -1),\, (7, 0, 0),\, (7, 0, -1),\, (0, 4, 0),\, (0, 4, -1),\, (7, 4, 0),\, (7, 4, -1)
    \right\}
    \]
    as displayed in \Cref{fig:ham-sandwich_polys}.
    \begin{figure}[ht]
        \centering
        \includegraphics[height=4.5cm]{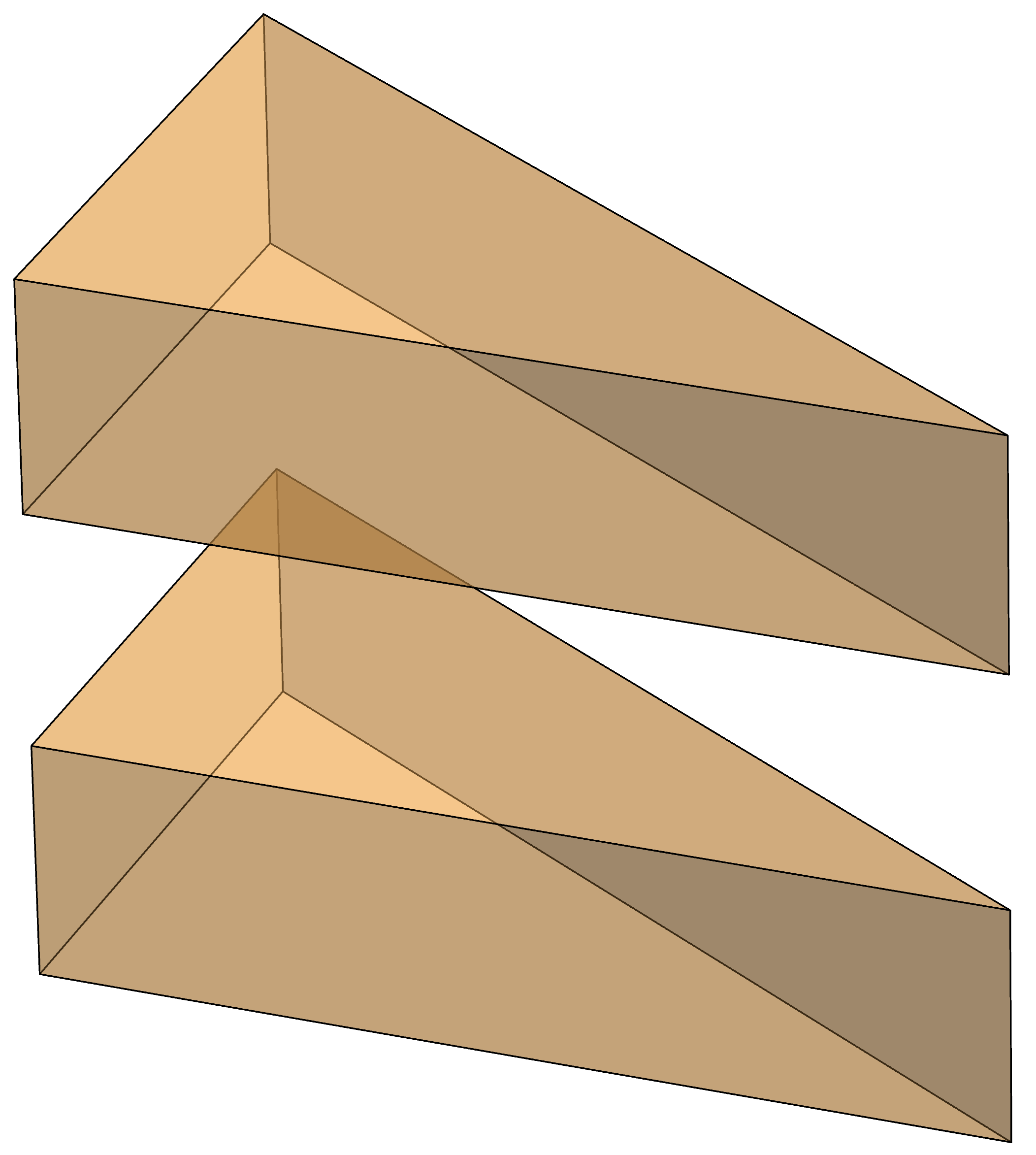}
        \qquad
        \includegraphics[height=4.5cm]{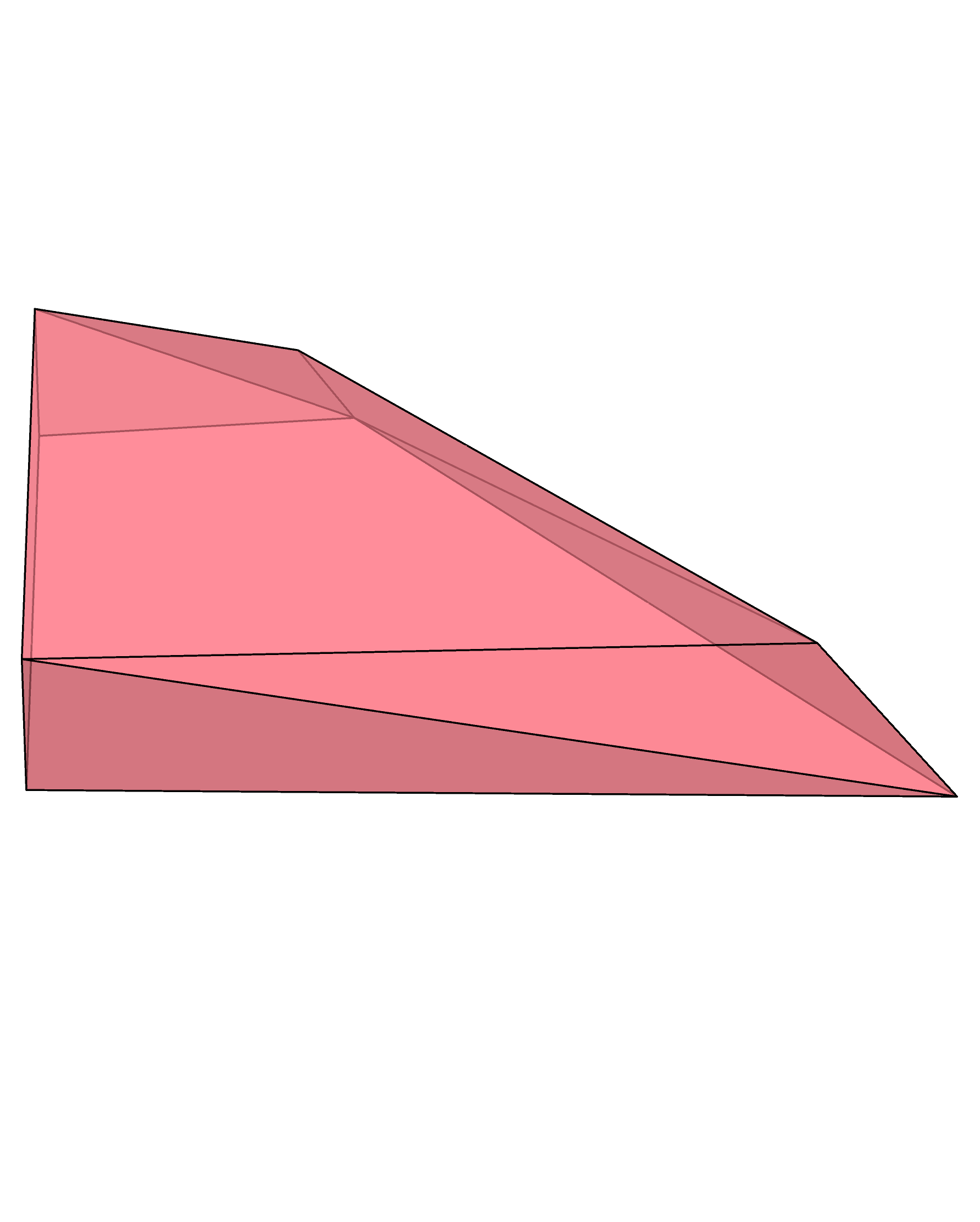}
        \qquad
        \includegraphics[height=4.5cm]{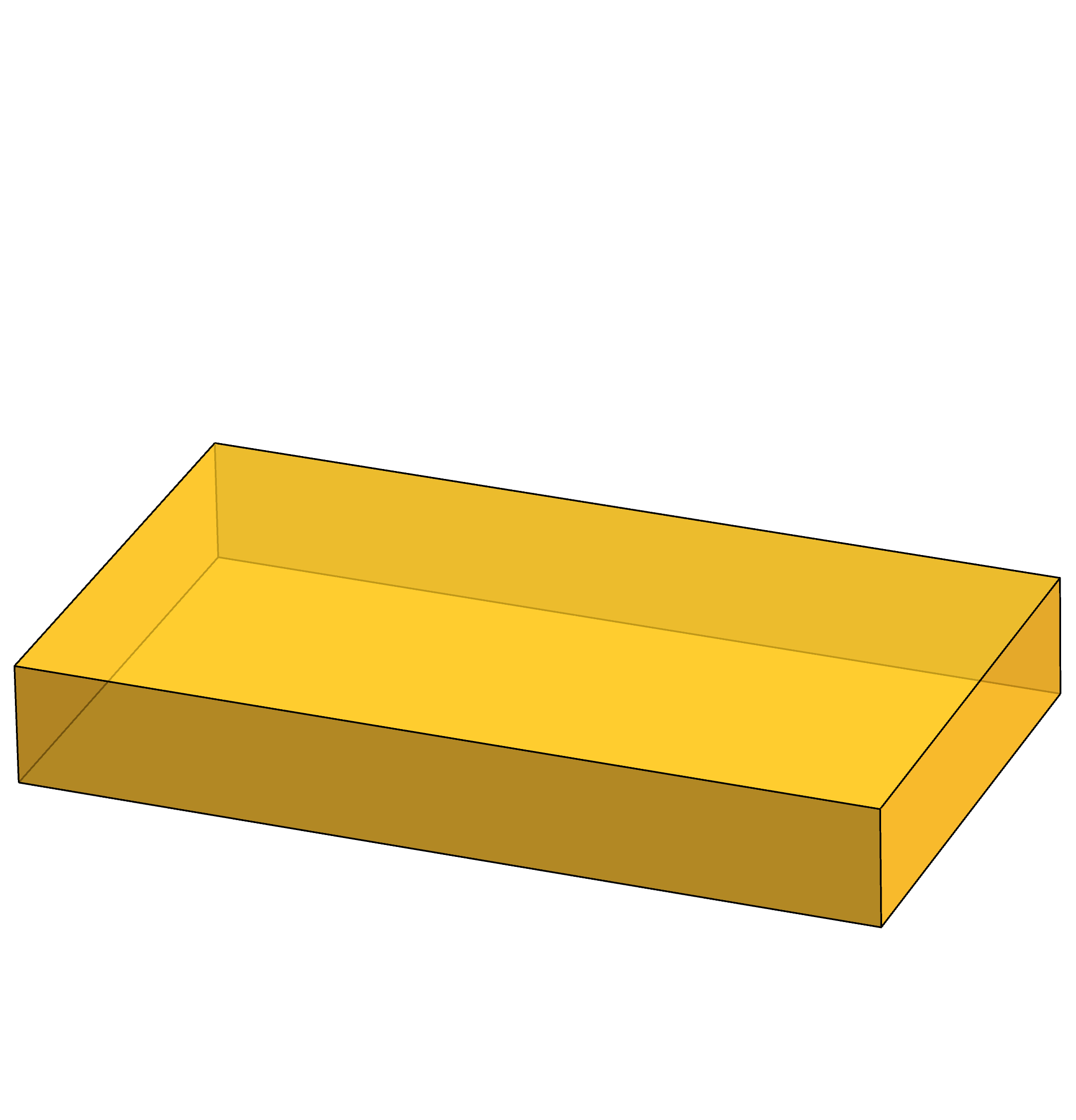}
        \caption{The polytopal bread, ham, and cheese, from \Cref{ex:ham-sandwich_3d}.}
        \label{fig:ham-sandwich_polys}
    \end{figure}
    Then, there exist exactly three hyperplanes that split each ingredient of the sandwich into two equal parts. 
    One of these cuts, in \Cref{fig:intro}, left, is the unique feasible point in the intersection $V_\frac{1}{2}({\rm bread}) \cap V_\frac{1}{2}({\rm ham}) \cap V_\frac{1}{2}({\rm cheese})$ restricted to the chamber
    {\small
    \begin{gather*}
    5u_2 - u_3 - t \le 0,\; 8u_1 - 3u_3 - t \le 0,\; -3u_3 + t \le 0, \\
    u_3 - t \le 0,\; -8u_1 + u_3 + t \le 0,\; -5u_2 - u_3 + t \le 0,
    \end{gather*}
    }
    where $V_\frac{1}{2}({\rm bread})$ is the set of points $(u,t)\in \Sph^2\times \R$ such that
    {\small
    \begin{gather*}
    \begin{multlined}
        4 t^3 u_1+t^2 (-60 u_1 u_2-12 u_1 u_3+15 u_2 u_3)+t \left(480 u_1^2 u_2-120 u_1 u_2 u_3+12 u_1 u_3^2\right) 
        -1280 u_1^3 u_2
        \\+1440 u_1^2 u_2 u_3-600 u_1 u_2^2 u_3-60 u_1 u_2 u_3^2-4 u_1 u_3^3+5 u_2 u_3^3 = 20 (3 u_1 u_2 u_3 (8 u_1-5 u_2)),
    \end{multlined}
    \end{gather*}}%
    and $V_\frac{1}{2}({\rm ham})$ and $V_\frac{1}{2}({\rm cheese})$ have similar expressions.
    The hyperplane corresponding to such a feasible point is then, up to four-digit precision, 
    \[
    0.5144 x + 0.2474 y + 0.8210 z = 1.8848.
    \]
    The associated slice of the sandwich is shown in \Cref{fig:ham-sandwich_slices}, left. The central and right slices correspond to the central and right cuts in \Cref{fig:intro}, and they are given respectively by the hyperplanes
    \begin{gather*}
        0.6043 x -0.3527 y + 0.7144 z = 1.0524, \\
        0.5127 x -0.7863 y + 0.3444 z = 0.0498.
        \qedhere
    \end{gather*}
    \begin{figure}[ht]
        \centering
        \includegraphics[height=4.5cm]{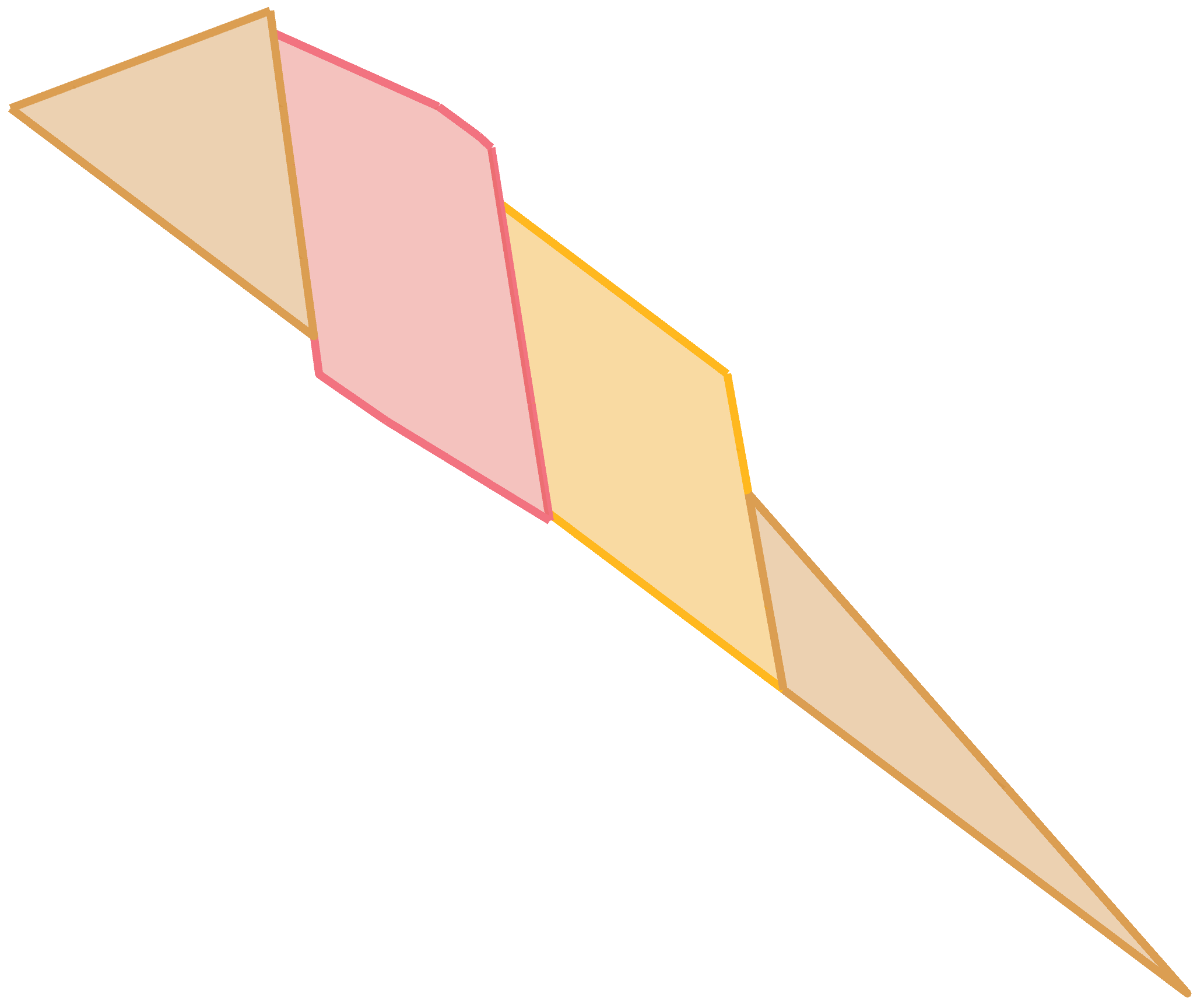}
        \hspace{-0.2cm}
        \includegraphics[height=4.5cm]{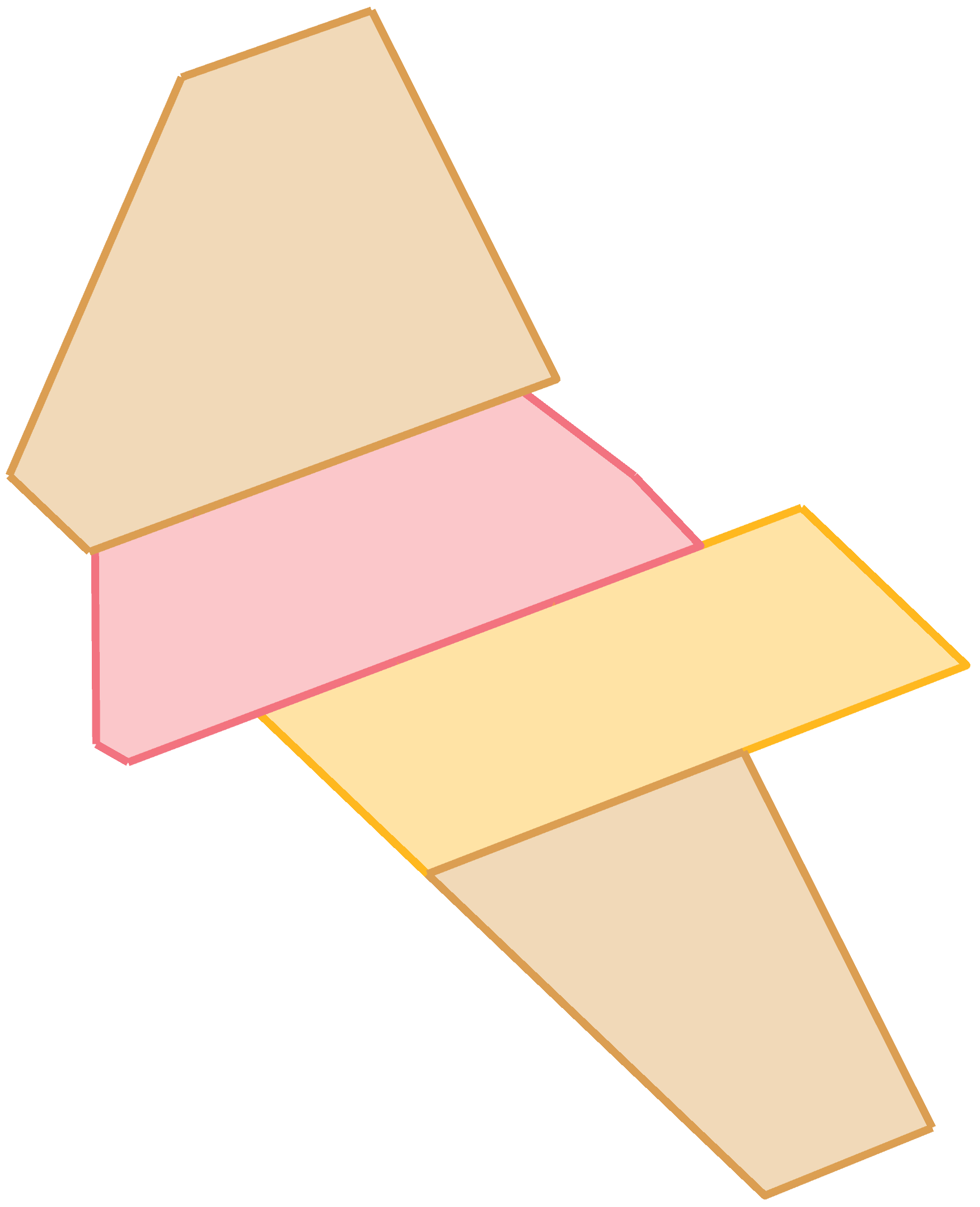}
        \quad
        \includegraphics[height=4.5cm]{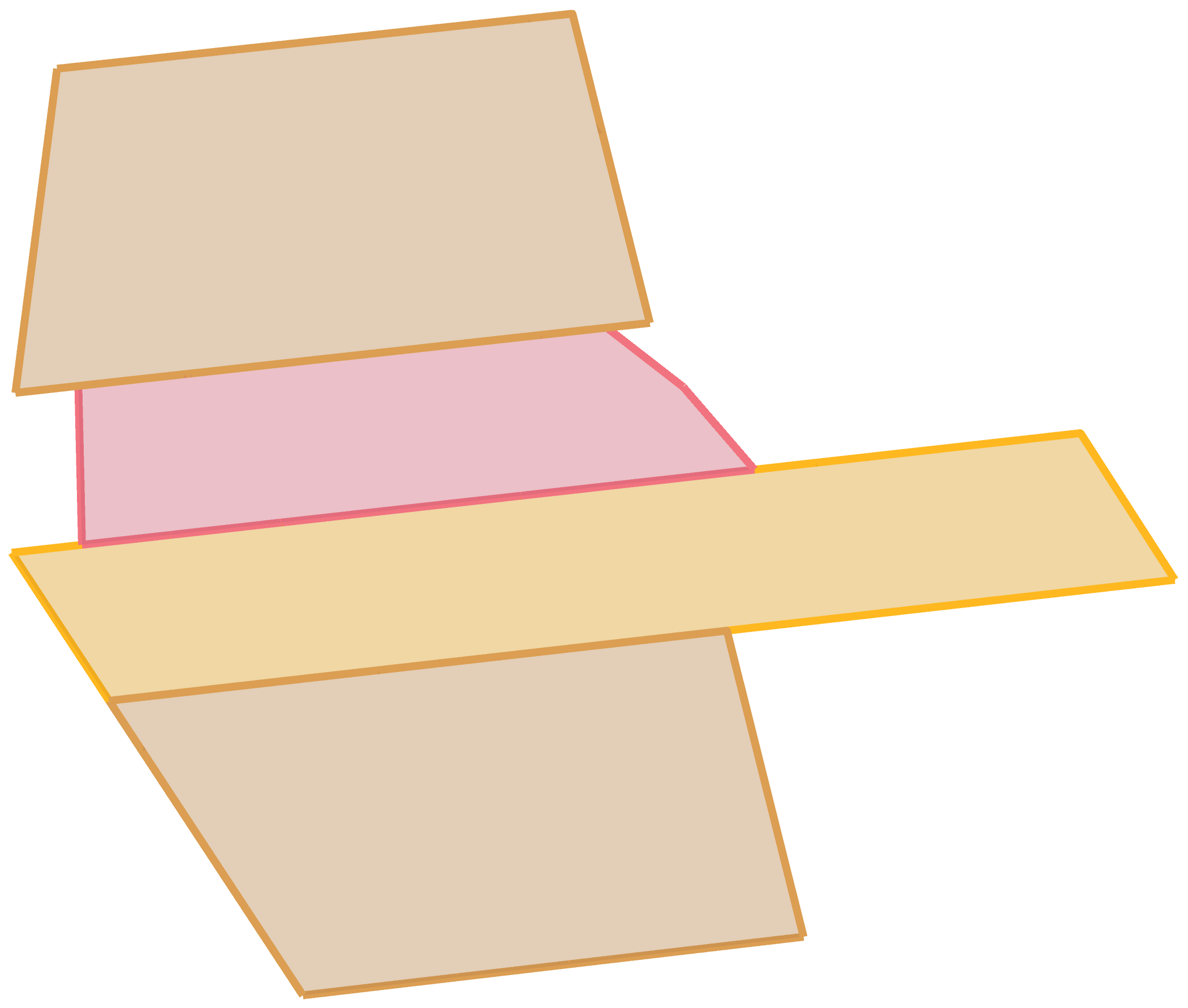}
        \caption{The three cuts of the ham-cheese-sandwich from \Cref{ex:ham-sandwich_3d}.}
        \label{fig:ham-sandwich_slices}
    \end{figure}
\end{example}

There is a straightforward generalization of \Cref{thm:HS_1/2_intro,thm:generalized_HS_intro} to measures with polynomial densities, as follows. 
Let $\mu=(\mu_1,\ldots,\mu_k)$ be a $k$-tuple of finite measures on $\R^d$
with densities $p_1,\ldots,p_k$, respectively. Assume that, for each $i\in[k]$, $\mu_i$ is supported on a finite union
$\bigcup_{j=1}^{m_i} Q_{i,j}$ of polytopes such that
\begin{equation}\label{eq:poly_densities}
p_i\big|_{Q_{i,j}} = p_{i,j} \in \mathbb Q[x_1,\ldots,x_d]
\end{equation}
for all $j\in[m_i]$, where $p_{i,j}\ge 0$ on $Q_{i,j}$.
Then, the proof of \Cref{thm:generalized_HS_intro} combined with \cite[Theorem 1.1 $(iii)$]{BDC2025bestslices} gives the following complexity statement. The corresponding algorithm follows the steps of \Cref{alg:alpha-HS}, with the cap-volume function replaced by the exact parametric integration of the polynomial density over the corresponding polytopal cap. See \cite{BDC2025bestslices} for further details.
\begin{corollary}
    Let $\mu=(\mu_1,\ldots,\mu_k)$ be a $k$-tuple of finite measures on $\R^d$
    with densities $p_i$ satisfying \eqref{eq:poly_densities}. Let $\alpha_1,\dots,\alpha_k \in (0,1) \cap \Q$. There exists an algorithm that takes as input the vertices of the polytopes $Q_{i,j}$ and the polynomials $p_{i,j}$ and outputs the set of all $\alpha$-cuts of $\mu$. For fixed dimension $d$, the complexity of the algorithm is polynomial in the input size, including the degrees and coefficient bit-size of the input polynomials.
\end{corollary}

\section{A semialgebraic interpolation from ham-sandwich cuts to centerpoints}\label{sec:transversal}

In the previous section, we studied cuts by halfspaces whose boundaries are $(d-1)$-dimensional hyperplanes. We now extend this point of view to lower-dimensional affine subspaces. Given a flat, i.e., an affine subspace $F\subset\mathbb R^d$ of any dimension, we measure its depth, i.e., the proportion of a set that must be contained in every halfspace containing $F$. If $F$ is a point, then this recovers the Tukey halfspace depth from the introduction. If $F$ is a hyperplane, then its depth records the smaller of the two cap volumes determined by $F$. In particular, $F$ is a ham-sandwich cut precisely when this depth is one half of the total volume for each measure. This framework allows an interpolation between the centerpoint theorem and the ham-sandwich theorem, which is provided from a theoretical point of view by the \emph{center transversal theorem}.

In this section, we develop a semialgebraic and computational framework for this interpolation. We first study spaces of affine subspaces satisfying prescribed depth inequalities, and show that these spaces are semialgebraic and effectively computable for polytopal measures. We then specialize to centerpoints of polytopes and relate the resulting semialgebraic sets to convex floating bodies, a classical construction in convex geometry.

\subsection{Semialgebraic spaces of deep flats and centerpoints}
The center transversal theorem gives a common generalization of the ham-sandwich theorem and the centerpoint theorem
\cite{Dolnikov92:GeneralizationHSCenterpoint,ZivVre90:ExtensionHS}; for a
modern reference and computational aspects, see \cite[Section~1.4]{Matousek02:LecturesDiscreteGeom} and \cite{SchniderBarequetWang19}. In this section, we extend our framework from the space of ham-sandwich cuts to the space of solutions to the center transversal theorem. As a special case, we obtain a description of the space of solutions to the centerpoint theorem.

\begin{theorem}[Center transversal theorem]\label{thm:center_transversal}
Let $1 \leq k \leq d$ and let $\mu_1,\ldots,\mu_k$ be finite Borel measures on $\mathbb R^d$ that vanish on hyperplanes. Then there exists an affine subspace $F\subset\mathbb R^d$ of dimension $k-1$ such that for every closed halfspace $H\subset\mathbb R^d$ with $F\subseteq H$ one has
\[
\mu_i(H)\ge \frac{1}{d-k+2}\,\mu_i(\mathbb R^d)
\qquad \text{for all } i=1,\ldots,k.
\]
\end{theorem}

For $k=d$, the subspace $F$ is a hyperplane. Since the measures vanish on hyperplanes, the two closed halfspaces bounded by $F$ both have at least one half of the total measure and hence \emph{exactly} one half. Thus, \Cref{thm:center_transversal} recovers the ham-sandwich theorem. For $k=1$, the subspace $F$ is a point, and \Cref{thm:center_transversal} recovers the centerpoint theorem.

In the case of polytopal measures, each $\mu_i$ is the restriction of the Lebesgue measure to a polytopal set $P_i$. Hence, the conclusion of the center transversal theorem can be written in terms of cap volumes as follows: there exists a $(k-1)$-dimensional affine subspace $F\subset\mathbb R^d$ such that, for every $(u,t)\in \Sph^{d-1}\times\mathbb R$ with $F\subseteq H(u\le t)$, one has
\[
\capvol_{P_i}(u,t)\ge
\frac{1}{d-k+2}\,\vol_d(P_i)
\qquad \text{for all } i=1,\ldots,k.
\]

\begin{definition}\label{def:flat-depth}
Let $S\subset\mathbb R^d$ be a full-dimensional compact set and let $F\subset\mathbb R^d$ be an affine subspace. The \emph{depth} of $F$ with respect to $S$ is
\[
\depth_S(F) = \inf \left\{ \capvol_S(u,t) \;\middle|\; (u,t)\in \Sph^{d-1}\times\mathbb R,\; F\subseteq H(u\le t) \right\}.
\]
\end{definition}

We now show that, for polytopal measures, the set of affine subspaces satisfying such depth inequalities is semialgebraic. Let $\graff(r,d)$ denote the affine Grassmannian, i.e., the space of all $r$-dimensional affine subspaces of $\mathbb R^d$.

\begin{theorem}[Semialgebraicity of the space of $\alpha$-deep flats]
\label{thm:deep-flats-semi-alg}
Let $P=Q_1\cup\cdots\cup Q_m\subset\mathbb R^d$ be a union of finitely many full-dimensional polytopes, and let $\alpha\in(0,1)$. For $1\le \ell\le d$, the set
\[
\mathcal F_\alpha^\ell(P) = \left\{ F\in\graff(\ell-1,d) \;\middle|\; \depth_P(F)\ge \alpha\,\vol_d(P) \right\}
\]
is semialgebraic.
\end{theorem}

\begin{proof}
We use the standard semialgebraic model of the affine Grassmannian. Namely, we write an affine $(\ell-1)$-dimensional subspace $F\subset\mathbb R^d$ as
\[
F=p+L,
\]
where $L\in\gr(\ell-1,d)$ is a linear subspace and $p\in L^\perp$. With this representation,
\[
\graff(\ell-1,d) = \{(L,p)\mid L\in\gr(\ell-1,d),\ p\in L^\perp\}
\]
is semialgebraic.
We first characterize the halfspaces containing $F = p + L$. For $u\in \Sph^{d-1}$, one has
\[
F\subseteq H(u\le t) \quad\Longleftrightarrow\quad u\in L^\perp \text{ and } \langle u,p\rangle\le t.
\]
Indeed, if $u\notin L^\perp$, then the map $x\mapsto\langle u,x\rangle$ is not constant in the affine space $p+L$, and therefore is unbounded on it. Hence, $F$ cannot be contained in any halfspace of the form $H(u\le t)$ with $u \not \in L^\perp$. If $u\in L^\perp$, then $\langle u,p+l\rangle=\langle u,p\rangle$ for all $l\in L$, so $F\subseteq H(u\le t)$ if and only if $\langle u,p\rangle\le t$.
Since $t\mapsto\capvol_P(u,t)$ does not decrease for fixed $u\in L^\perp$, the infimum over all admissible $t$ is attained at $t=\langle u,p\rangle$. Therefore, $\depth_P(p+L)\ge \alpha\,\vol_d(P)$ is equivalent to
\[
\capvol_P(u,\langle u,p\rangle)\ge \alpha\,\vol_d(P) \qquad \text{for all } u\in \Sph^{d-1}\cap L^\perp.
\]
Thus
\[
\mathcal F_\alpha^\ell(P) = \left\{ (p,L)\in\graff(\ell-1,d) \;\middle|\; \forall u\in \Sph^{d-1}\cap L^\perp,\  \capvol_P(u,\langle u,p\rangle)\ge \alpha\,\vol_d(P) \right\}.
\]
By \Cref{lem:cap-volume-semialg}, the cap-volume function $\capvol_P$ is semialgebraic. The conditions $u\in \Sph^{d-1}$ and $u\in L^\perp$ are given by polynomial equations, and the map $(p,u)\mapsto\langle u,p\rangle$ is polynomial.
Hence the relation
\[
\capvol_P(u,\langle u,p\rangle)\ge \alpha\,\vol_d(P)
\]
is semialgebraic in $(p,L,u)$. By quantifier elimination, $\mathcal F_\alpha^\ell(P)$ is semialgebraic as well.
\end{proof}

As a corollary of \Cref{thm:deep-flats-semi-alg}, we get that the set of simultaneous $\alpha$-deep flats of finitely many polytopal measures is semialgebraic.
\begin{corollary}\label{cor:simultaneous-deep-flats-semi-alg}
Let $\mathcal P=(P_1,\ldots,P_k)$ be a $k$-tuple of polytopal measures in $\mathbb R^d$, and let $\alpha=(\alpha_1,\ldots,\alpha_k)\in(0,1)^k$. Then, for every $\ell \in [d]$,
\[
\mathcal F_\alpha^\ell(\mathcal P) = \bigcap_{i=1}^k \mathcal F_{\alpha_i}^\ell(P_i) \subseteq \graff(\ell-1,d)
\]
is semialgebraic.
\end{corollary}

The preceding corollary is effective: using the explicit semialgebraic description of the cap-volume functions from \Cref{lem:cap-volume-semialg}, one can compute defining equations and inequalities for these sets. Next, we make these computational aspects more explicit.

\begin{proposition}\label{prop:deep-flats-algorithm}
Let $\mathcal P=(P_1,\ldots,P_k)$ be a $k$-tuple of rational polytopal measures in $\mathbb R^d$, and let $\alpha=(\alpha_1,\ldots,\alpha_k)\in(0,1)^k \cap \Q^k$.
For fixed $d$ and $\ell$, the set $\mathcal F_\alpha^\ell(\mathcal P)$ can be computed in polynomial time.
\end{proposition}

\begin{proof}
By \Cref{lem:cap-volume-semialg}, the cap-volume functions appearing in the definition of $\mathcal F_\alpha^\ell(\mathcal P)$ have semialgebraic descriptions of polynomial size, when the dimension $d$ is fixed. The affine Grassmannian $\graff(\ell-1,d)$ has fixed dimension, and the defining formula for $\mathcal F_\alpha^\ell(\mathcal P)$ uses a fixed number of quantified variables when $d$ and $\ell$ are fixed.
Therefore, quantifier elimination and standard algorithms for semialgebraic sets compute a quantifier-free description, decide emptiness, and produce sample points in polynomial time in the input size.
\end{proof}

\begin{proof}[Proof of \Cref{thm:centerpoints_intro}]
    For $\ell=1$, the points $(p,L) \in \graff(0,d)$ with $L \in \gr(0,d)$, $p \in L^\perp$ (as in the proof of \Cref{thm:deep-flats-semi-alg}) necessarily satisfies that $L$ is the trivial linear space of dimension $0$, and $p \in \R^d$. We can thus identify $\graff(0,d)$ with $\R^d$.
    To obtain the centerpoints, we consider \Cref{prop:deep-flats-algorithm} for $k = 1$ and $\alpha = (\frac{1}{d+1})$. Then
    $
        \mathcal F_{\frac{1}{d+1}}^1(P)
    $
    is precisely the set of all centerpoints,
    and is non-empty by the centerpoint theorem.
    By \Cref{prop:deep-flats-algorithm}, this set can be computed in polynomial time.
    Algorithmically, one first computes the piecewise rational representation of $\capvol_P$ (see Steps 1--4 in \Cref{alg:alpha-HS}), applies quantifier elimination to obtain a quantifier-free semialgebraic description of $\mathcal F_{\frac{1}{d+1}}^1(P)$, and computes a sample point of this semialgebraic set.
\end{proof}

\subsection{Centerpoints of polytopes via floating bodies} \label{sec:floating-bodies}
We now focus on the second extremal case of the center transversal theorem, namely the computation of centerpoints. In this subsection we restrict our discussion to a \emph{single full-dimensional convex polytope} $P\subset\mathbb R^d$, and show that the set of centerpoints of $P$ is a floating body of $P$.
We use the term \emph{floating body} for what is often called the convex floating body, in order to distinguish it from Dupin's floating body; see \cite{BarLar88:Floating,SchWer90:Floating,MorWer19:FloatingPolytopes}. \Cref{fig:triangle-floating-body} shows the floating body of a triangle (computed with methods from this subsection).

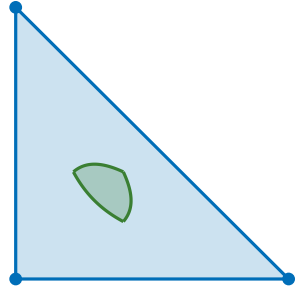
\begin{figure}[hbt]
    \centering
    \begin{tikzpicture}[scale=1.2]

\definecolor{tri_color}{named}{NavyBlue}
\definecolor{float_color}{named}{OliveGreen}

\begin{scope}[scale=1]
    \clip (-1.2,-1.2) rectangle (2.2,2.2);

    \filldraw[
      fill=tri_color,
      fill opacity=0.2,
      draw=tri_color,
      very thick
    ]
    (-1,-1) -- (-1,2) -- (2,-1) -- cycle;

    \foreach \p in {(-1,-1), (-1,2), (2,-1)}{
      \fill[tri_color] \p circle (2pt);
    }

\pgfmathsetmacro{\a}{(1 - sqrt(3))/4}
\pgfmathsetmacro{\b}{(-1 + sqrt(3))/2}

\fill[float_color, opacity=0.25]
  plot[samples=120, smooth, domain=\a:\b, variable=\x]
    ({-\x}, {-(1 - 4*\x*\x)/(4*\x - 4)})
  --
  plot[samples=120, smooth, domain=\b:\a, variable=\y]
    ({-(1 - 4*\y*\y)/(4*\y - 4)}, {-\y})
  --
  plot[samples=120, smooth, domain=\b:\a, variable=\x]
    ({-\x}, {-(4*\x - 1)/(4*\x - 4)})
  -- cycle;

\draw[very thick, float_color, samples=120, smooth, domain=\a:\b, variable=\x]
  plot ({-\x}, {-(4*\x - 1)/(4*\x - 4)});

\draw[very thick, float_color, samples=120, smooth, domain=\a:\b, variable=\x]
  plot ({-\x}, {-(1 - 4*\x*\x)/(4*\x - 4)});

\draw[very thick, float_color, samples=120, smooth, domain=\a:\b, variable=\y]
  plot ({-(1 - 4*\y*\y)/(4*\y - 4)}, {-\y}); 
  
\end{scope}
\end{tikzpicture}
    \caption{The triangle $P$ and its floating body $P_{\frac{1}{3}}$, as computed in \Cref{ex:floating-body} using the method of this section.}
    \label{fig:triangle-floating-body}
\end{figure}

\begin{definition}
Let $K\subset\mathbb R^d$ be a convex body and let $\alpha\in(0,1)$. The \emph{floating body} $K_\alpha$ is
\[
K_\alpha = \bigcap_{\substack{u\in \Sph^{d-1},\, t\in\mathbb R\\ \capvol_K(u,t)\le \alpha\,\vol_d(K)}} H(u\ge t).
\]
\end{definition}
Equivalently, since $t\mapsto\capvol_K(u,t)$ is nondecreasing, and since $K$ is convex and full-dimensional, the intersection above may be taken over the boundary values where the cap-volume function is equal to $\alpha\,\vol_d(K)$.
Thus
\[
K_\alpha = \bigcap_{(u,t)\in V_\alpha(K)} H(u\ge t).
\]
Replacing $(u,t)$ by $(-u,-t)$ gives the equivalent form
\[
K_\alpha = \bigcap_{(u,t)\in V_{1-\alpha}(K)} H(u\le t).
\]
This last description connects floating bodies directly to the quantile loci studied in the previous section. We now use it to show that, for polytopes, floating bodies remain within the semialgebraic category.

\begin{proposition}\label{prop:semialg_floating}
Let $P\subset\mathbb R^d$ be a full-dimensional polytope and let $\alpha\in(0,1)$. Then the floating body $P_\alpha$ is semialgebraic.
\end{proposition}

\begin{proof}
Using the previous description, we can write
\[
P_\alpha = \{x\in\mathbb R^d \mid \langle u,x\rangle\le t \text{ for all } (u,t)\in V_{1-\alpha}(P)\}.
\]
Equivalently, $x\in P_\alpha$ if and only if for all $(u,t)\in \Sph^{d-1}\times\mathbb R$ 
\[
\capvol_P(u,t)=(1-\alpha)\vol_d(P) \Rightarrow \langle u,x\rangle\le t.
\]
By \Cref{prop:quantile-locus-semi-alg}, the set $V_{1-\alpha}(P)$ is semialgebraic. Hence, by quantifier elimination, $P_\alpha$ is semialgebraic.
\end{proof}

We next identify the floating body parameter that corresponds to centerpoints. The number $1/(d+1)$ is exactly the depth threshold appearing in the centerpoint theorem.

\begin{proposition}\label{prop:centerpoints-floating}
Let $P\subset\mathbb R^d$ be a full-dimensional convex body. Then the set of centerpoints of $P$ is precisely $P_{\frac{1}{d+1}}$.
\end{proposition}

\begin{proof}
For $x\in\mathbb R^d$, the halfspace depth of $x$ with respect to $P$ is
\[
\depth_P(x) = \inf_{u\in \Sph^{d-1}}\capvol_P(u,\langle u,x\rangle).
\]
Indeed, among all halfspaces of the form $H(u\le t)$ containing $x$, the smallest cap volume is obtained when $t=\langle u,x\rangle$, because $t\mapsto\capvol_P(u,t)$ is nondecreasing.
Thus $x$ is a centerpoint if and only if
\[
\capvol_P(u,\langle u,x\rangle) \ge \frac{1}{d+1}\vol_d(P) \qquad \text{for all } u\in \Sph^{d-1}.
\]
By the definition of the floating body, this is equivalent to $x\in P_{1/(d+1)}$.
\end{proof}

Combining the previous two propositions gives an effective description of the set of centerpoints of a rational polytope. This is the centerpoint analogue of the ham-sandwich algorithm developed above. We showcase our method in two examples.

\begin{example}\label{ex:floating-body}
Consider the triangle $P$ from \Cref{ex:triangle}. Its floating body $P_{\frac{1}{3}}$, equivalently its set of centerpoints, is the semialgebraic set defined by
\[
4xy+4x+4y+1\ge 0,\quad
4x^2+4xy+4y-1\le 0,\quad
4xy+4y^2+4x-1\le 0.
\]
This description can be obtained from $V_{\frac{2}{3}}(P)$ by quantifier elimination. We performed this in \texttt{Mathematica} in $0.2$ seconds, and plotted the sets in \Cref{fig:triangle-floating-body}.
\end{example}

\begin{figure}
    \centering
    \includegraphics[height=5cm]{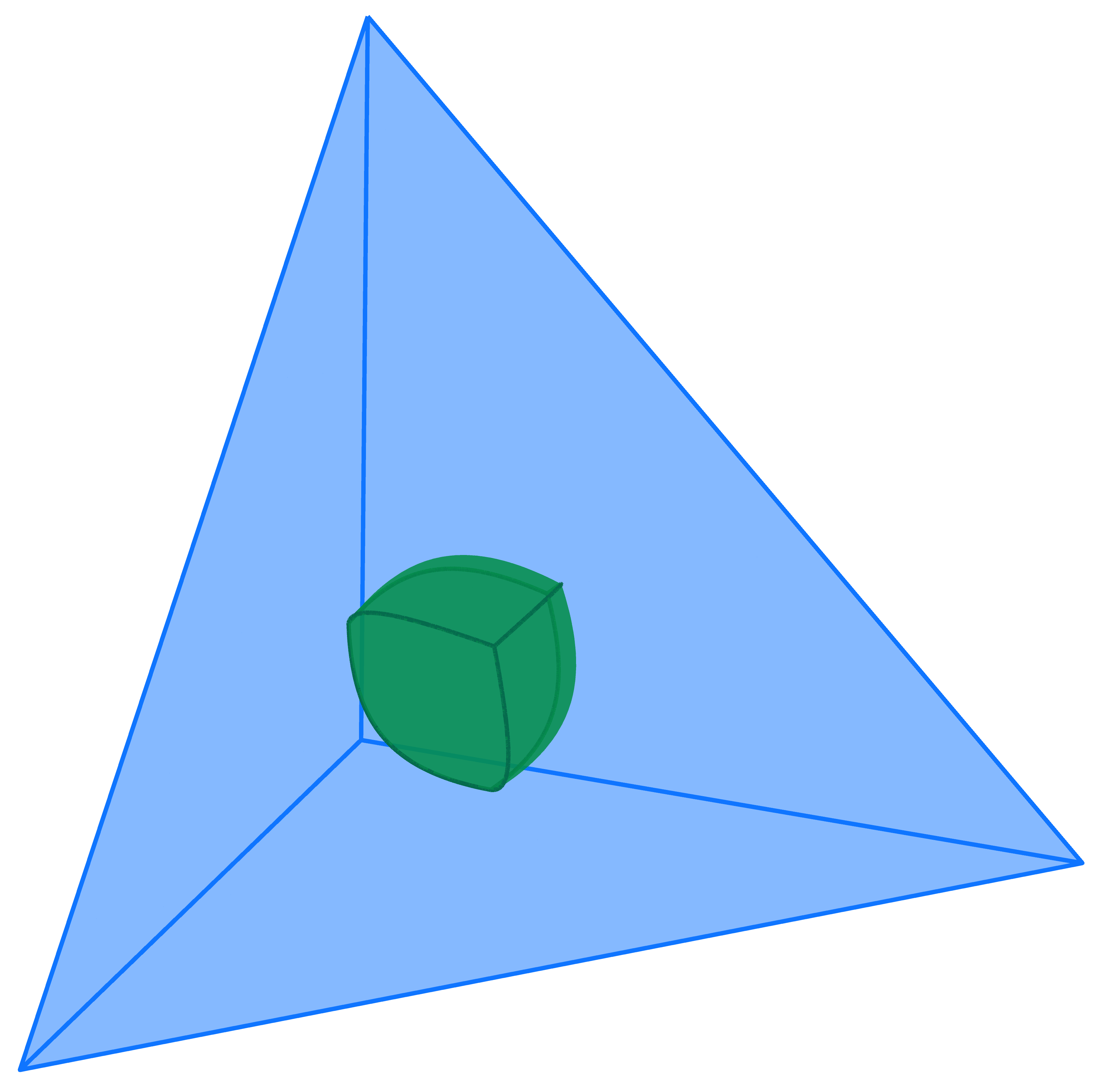}
    \qquad
    \qquad
    \includegraphics[height=5cm]{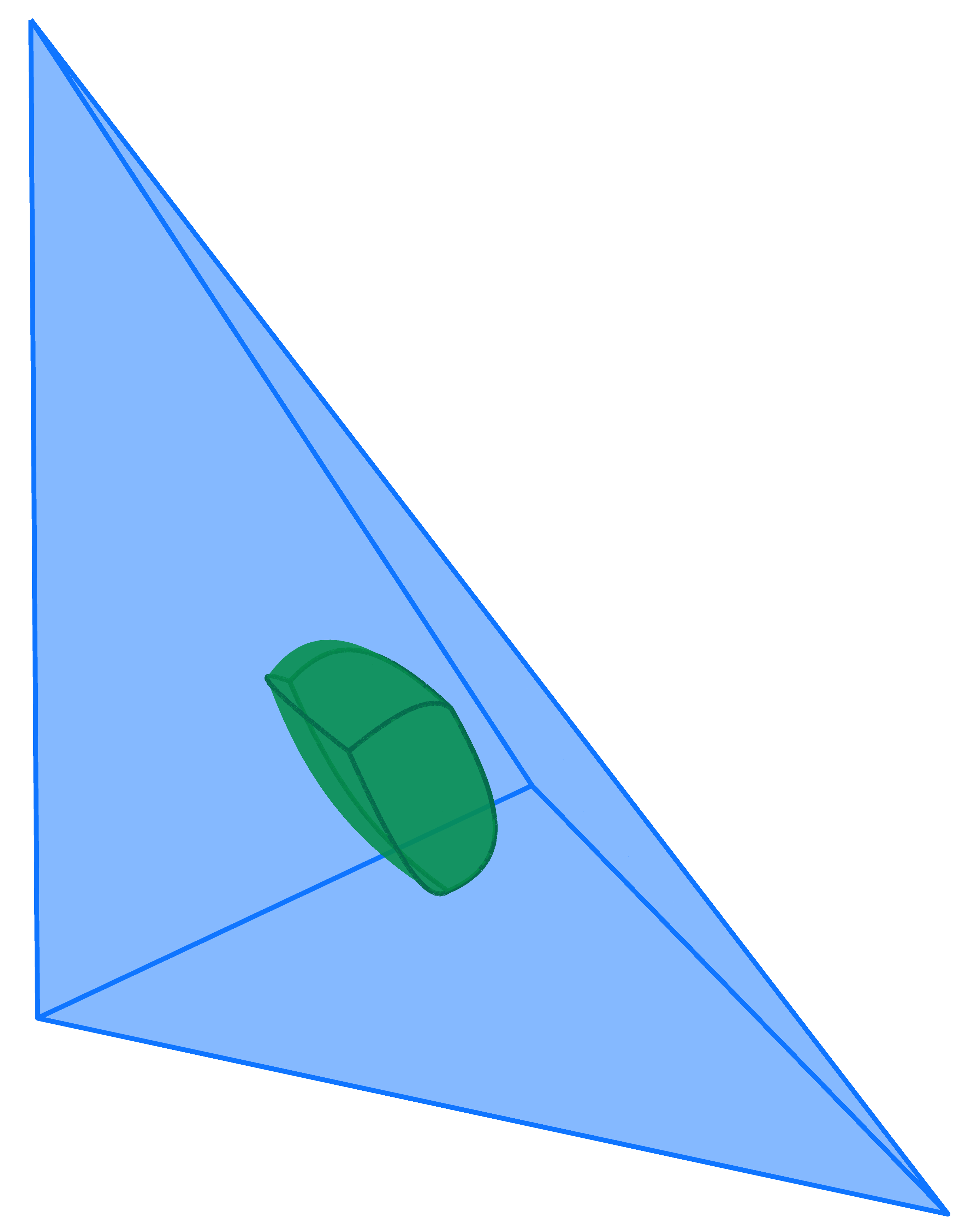}
    \caption{Two points of view of the blue tetrahedron $P$ and its green floating body $P_{\frac{1}{4}}$, from \Cref{ex:deep-lines-two-tetrahedra}.}
    \label{fig:float_body_3d}
\end{figure}

Finally, we remark that the space of $\alpha$-deep flats of $\mathcal{P}=(P_1,\ldots,P_k)$, when the $P_i$ are all convex polytopes, can be interpreted using floating bodies also when $k>1$. This is made explicit in the following example.
\begin{example}\label{ex:deep-lines-two-tetrahedra}
Consider the tetrahedra
\[
P
=
\conv\{(0,0,0),(6,0,0),(0,6,0),(0,0,6)\}
\]
and
\[
Q =
\conv\{(3,3,3),(-3,3,3),(3,-3,3),(3,3,-3)\}.
\]
We describe the set of lines of depth at least $\frac14$ with respect to both tetrahedra, namely $\mathcal F^\ell_{\alpha}(P,Q)$ for $\ell = 2$, $\alpha_1 = \alpha_2=\frac14$. Working in the affine chart of $\graff(1,3)$ consisting of lines whose direction has nonzero third coordinate, we write such an affine line as
\[
L_{m,n,p,q}
=
\{(p+mt,q+nt,t) \mid t\in\R\}.
\]
Consider the $\frac14$-floating bodies of $P$ and $Q$, namely the sets of their centerpoints. $P_{\frac14}$ is the semialgebraic set defined by $(x,y,z)\in P$ and 
\[
xyz \ge 2,\,
yz(6-x-y-z) \ge 2,\,
xz(6-x-y-z) \ge 2,\, xy(6-x-y-z) \ge 2.
\]
\Cref{fig:float_body_3d} displays $P$ in blue and $P_{\frac14}$ in green. Analogously, $Q_{\frac14}$ consists of points $(x,y,z)\in Q$ such that
\begin{gather*}
(3-x)(3-y)(3-z)\ge 2,\,
(x+y+z-3)(3-y)(3-z)\ge 2, \\ (x+y+z-3)(3-x)(3-z)\ge 2,\,
(x+y+z-3)(3-x)(3-y)\ge 2.
\end{gather*}
A line has depth at least $\frac14$ for both $P$ and $Q$ if and only if it intersects both floating bodies. Therefore, in the affine chart above, the set of such lines is the semialgebraic set
\[
\begin{aligned}
\mathcal F^2_{\frac14}(P,Q)
=
\Bigl\{
(m,n,p,q)\in\R^4 \mid
 L_{m,n,p,q}\cap P_\frac14 \neq \emptyset
\text{ and }
L_{m,n,p,q} \cap Q_\frac14 \neq \emptyset
\Bigr\}.
\end{aligned}
\]
This description uses two existential quantifiers hidden in the non-empty conditions. 

\begin{figure}[ht]
    \centering
    \begin{tikzpicture}[scale=1.5]
\definecolor{line_color}{named}{OliveGreen}
\begin{scope}
    \clip (0.6,0.6) rectangle (2.4,2.4);

    \pgfmathsetmacro{\pa}{8/9}
    \pgfmathsetmacro{\pb}{2}
    \pgfmathsetmacro{\pc}{19/9}
    \pgfmathsetmacro{\pd}{1}

    \fill[line_color, opacity=0.25]
      plot[samples=120, smooth, domain=\pa:\pb, variable=\x]
        ({\x}, {6 - \x - sqrt(8/\x)})
      --
      plot[samples=120, smooth, domain=\pb:\pa, variable=\y]
        ({6 - \y - sqrt(8/\y)}, {\y})
      --
      plot[samples=120, smooth, domain=\pc:\pd, variable=\x]
        ({\x}, {- \x + sqrt(8/(3-\x))})
      --
      plot[samples=120, smooth, domain=\pd:\pc, variable=\y]
        ({- \y + sqrt(8/(3-\y))}, {\y})
      -- cycle;

    \draw[very thick, line_color, samples=120, smooth, domain=\pa:\pb, variable=\x]
      plot ({\x}, {6 - \x - sqrt(8/\x)});

    \draw[very thick, line_color, samples=120, smooth, domain=\pb:\pa, variable=\y]
      plot ({6 - \y - sqrt(8/\y)}, {\y});

    \draw[very thick, line_color, samples=120, smooth, domain=\pc:\pd, variable=\x]
      plot ({\x}, {- \x + sqrt(8/(3-\x))});

    \draw[very thick, line_color, samples=120, smooth, domain=\pd:\pc, variable=\y]
      plot ({- \y + sqrt(8/(3-\y))}, {\y});
\end{scope}
\end{tikzpicture}
    \caption{Set of vertical lines in $\mathcal F^2_{\frac14}(P,Q)$, from \Cref{ex:deep-lines-two-tetrahedra}.}
    \label{fig:vertical_deep_lines}
\end{figure}
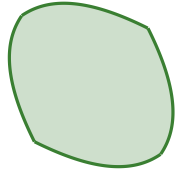
By restricting ourselves to vertical lines, namely to the two-dimensional subspace of $\R^4$ with $m=n=0$, we can derive the quantifier-free description of the set of vertical lines that belong to $\mathcal F^2_{\frac14}(P,Q)$. This is 
\[
\{(p,q)\in \R^2 \mid p (6 - p - q)^2 \ge 8, q (6 - p - q)^2 \ge 8, (3 - p) (p + q)^2 \ge 8, (3 - q) (p + q)^2 \ge 8 \},
\]
and it is displayed in \Cref{fig:vertical_deep_lines}. 
In particular, one can check that at the point $(0,0,3/2,3/2)$ the inequalities defining $\mathcal F^2_{\frac14}(P,Q)$ are all strict, hence the latter is full-dimensional.
\end{example}

\section{Conclusions}

We presented a methodology for solving several fair or proportional partition problems when we partition polytopal measures using one hyperplane cut. Our method expresses cutting conditions as semialgebraic constraints and then uses the power of algebraic computation to give polynomial-time algorithms (for fixed dimension). We provide efficient semialgebraic descriptions of ``quantile bodies'' or ``$\alpha$-partition loci'' with controlled degree and number of inequalities in terms of the complexity of the input polytopes. In our setup of exact algebraic computation one can ask for a rational/algebraic description of a bisecting hyperplane and receive verifiable certificates as well. Despite this, several open problems remain:
\begin{itemize}
\item Are there other families of measures that are particularly amenable to computation? Here we showed that partitioning polytopal measures is semialgebraic and thus efficiently computable. Since polytopes can approximate other convex bodies, is the extension by approximation efficient?

\item Practical computation was possible in small instances. A natural future direction is to obtain faster practical algorithms for polytopal ham-sandwich cuts or centerpoints. If approximation is sufficient, one can investigate quantified numerical approximations of the semialgebraic sets presented here.

\item A different generalization of these problems asks for which triples $(d,m,k)$ $m$ polytopal measures in $\mathbb{R}^d$ can be equipartitioned by $k$ hyperplanes. Here we solve the case $k=1$;  the case $k>1$ is challenging. Concretely, what is the computational complexity of finding such partitions when existence is guaranteed? For various hardness results in this direction, see \cite{Schnider22}.

\item Given a polytope, it is also interesting to study the set of points lying on hyperplanes that cut off a prescribed fraction of volume, or the envelope of those hyperplanes. In a forthcoming paper we will consider the geometry of $\alpha$-quantile bodies for polytopes and the envelopes of those hyperplanes.
\end{itemize}

\printbibliography

\vskip 1cm
\noindent
\textsc{Marie-Charlotte Brandenburg}\\
\textsc{Ruhr-Universit\"at Bochum}\\
\url{marie-charlotte.brandenburg@rub.de}\\

\noindent
\textsc{Jesús A. De Loera}\\
\textsc{University of California, Davis}\\
\url{jadeloera@ucdavis.edu}\\

\noindent
\textsc{Chiara Meroni}\\
\textsc{ETH Institute for Theoretical Studies}\\
\url{chiara.meroni@eth-its.ethz.ch}

\end{document}